\documentclass[a4paper,UKenglish]{article}
\usepackage{mathtools}

\usepackage[notextcomp]{stix}
\usepackage{graphicx} 
\usepackage{amsmath,amsthm,amsfonts}
\usepackage{tikz}
\usepackage{todonotes}
\usepackage{xcolor}
\usepackage{float}
\usepackage{hyperref}
\usepackage{cleveref, thm-restate}

\usepackage[left=2.5cm,right=2.5cm,top=3cm,bottom=3cm]{geometry}

\usetikzlibrary{automata, arrows.meta, positioning}

\newtheorem{theorem}{Theorem}
\newtheorem{problem}{Problem}
\newtheorem{lemma}{Lemma}
\newtheorem{definition}{Definition}
\newtheorem{example}{Example}

\renewcommand{\alph}{\mathrm{alph}}
\newcommand{\alphabet}{\mathrm{alph}}
\newcommand{\dfawtl}{{\textrm{TFA}}}
\newcommand{\dfa}{{\textrm{DFA}}}

\newcommand{\nfa}{{\textrm{NFA}}}
\newcommand{\HPP}{{\textrm{HPP}}}
\newcommand{\ETH}{{\textrm{ETH}}}
\newcommand{\REG}{{\textrm{REG}}}
\newcommand{\CF}{{\textrm{CF}}}
\newcommand{\CFL}{{\textrm{CFL}}}
\newcommand{\CFG}{{\textrm{CFG}}}
\newcommand{\CS}{{\textrm{CS}}}
\newcommand{\CSL}{{\textrm{CSL}}}
\newcommand{\CSG}{{\textrm{CSG}}}
\newcommand{\CNF}{{\textrm{CNF}}}
\newcommand{\bigo}{{\mathcal O}}

\newcommand{\FPT}{{\textrm{FPT}}}
\newcommand{\SAS}{{\textrm{SAS}}}
\newcommand{\PCP}{{\textrm{PCP}}}
\newcommand{\SLP}{{\textrm{SLP}}}

\def\N{\mathbb{N}}

\newcommand{\Down}[1]{#1\hspace{-1mm}\downarrow}
\newcommand{\Up}[1]{#1\hspace{-1mm}\uparrow}
\newcommand{\len}[1]{|#1|}

\newcommand{\iotaA}{\iota_\forall}
\newcommand{\iotaE}{\iota_\exists}
\newcommand{\p}{\prime}
\newcommand{\poly}{\mathrm{poly}}

\title{Subsequence Matching and Analysis Problems for Formal Languages} 

\author{%
Szil\'{a}rd  Zsolt Fazekas%
\thanks{Akita University, Japan,
 \texttt{szilard.fazekas@ie.akita-u.ac.jp}, 
 {Supported by the JSPS Kakenhi Grant Number 23K10976.}}
\and
Tore Ko{\ss}
\thanks{University of G\"ottingen, Germany,
 \texttt{tore.koss@uni-goettingen.de}}
\and
Florin Manea
\thanks{University of G\"ottingen, Germany,
 \texttt{florin.manea@cs.uni-goettingen.org}, 
 {Supported by the German Research Foundation (Deutsche Forschungsgemeinschaft, DFG) in the framework of the Heisenberg Programme project number 466789228.}}
\and
Robert {Merca\c s}
\thanks{Loughborough University, UK,
\texttt{r.g.mercas@lboro.ac.uk}, 
{Supported by a Return Fellowship from the Alexander von Humboldt Foundation for a research stay at University of Göttingen, Germany.}}
\and
Timo Specht
\thanks{University of G\"ottingen, Germany,
\texttt{timo.specht@stud.uni-goettingen.de}}
}

\date{}

\begin{document}

\maketitle

\begin{abstract}
In this paper, we study a series of algorithmic problems related to the subsequences occurring in the strings of a given language, under the assumption that this language is succinctly represented by a grammar generating it, or an automaton accepting it. In particular, we focus on the following problems: Given a string $w$ and a language $L$, does there exist a word of $L$ which has $w$ as subsequence? Do all words of $L$ have $w$ as a subsequence? Given an integer $k$ alongside $L$, does there exist a word of $L$ which has all strings of length $k$, over the alphabet of $L$, as subsequences? Do all words of $L$ have all strings of length $k$ as subsequences? For the last two problems, efficient algorithms were already presented in [Adamson et al., ISAAC 2023] for the case when $L$ is a regular language, and efficient solutions can be easily obtained for the first two problems. We extend that work as follows: we give sufficient conditions on the class of input-languages, under which these problems are decidable; we provide efficient algorithms for all these problems in the case when the input language is context-free; we show that all problems are undecidable for context-sensitive languages. Finally, we provide a series of initial results related to a class of languages that strictly includes the regular languages and is strictly included in the class of context-sensitive languages, but is incomparable to the of class context-free languages; these results deviate significantly from those reported for language-classes from the Chomsky hierarchy.\\

\end{abstract}

 \section{Introduction}\label{sec:intro}
A string $v$ is a subsequence of a string $w$, denoted $v\leq w$ in the following, if there exist (possibly empty) strings  $x_1, \ldots, x_{\ell+1}$ and $v_1, \ldots,  v_\ell$ such that $v = v_1 \cdots v_\ell$ and $w = x_1 v_1 \cdots x_\ell v_\ell x_{\ell+1}$. In other words, $v$ can be obtained from $w$ by removing some of its letters. 

The concept of subsequence appears and plays important roles in many different areas of theoretical computer science. Prime examples are the areas of  combinatorics on words, formal languages, automata theory, and logics, where subsequences are studied in connection to piecewise testable languages~\cite{simonPhD,Simon72,KarandikarKS15,journals/lmcs/KarandikarS19}, in connection to subword-order and downward-closures~\cite{HalfonSZ17,KuskeZ19,Kuske20,Zetzsche16,Zetzsche18,AnandZ23}), in connection to binomial equivalence, binomial complexity, or to subword histories ~\cite{RigoS15,FreydenbergerGK15,LeroyRS17a,Rigo19,Seki12,Mat04,Salomaa05}. Subsequences are important objects of study also in the area of algorithm-design and complexity; to this end, we mention some classical algorithmic problems such as the computation of {\em longest common subsequences} or of the {\em shortest common supersequences}~\cite{chvatal,Hirschberg77,HuntS77,Maier:1978,MasekP80,NakatsuKY82,DBLP:journals/tcs/Baeza-Yates91,BergrothHR00}, the testing of the Simon congruence of strings and the computation of the arch factorisation and universality of strings~\cite{TCS::Hebrard1991,garelCPM,SimonWords,DBLP:conf/wia/Tronicek02,CrochemoreMT03,dlt2020,day2021edit,KufMFCS,GawrychowskiEtAl2021,KoscheKMS21,FleischmannKKMNSW23}; see also~\cite{SurveyNCMA} for a survey on combinatorial pattern matching problems related to subsequences. Moreover, these algorithmic problems and other closely related ones have recently regained interest in the context of fine-grained complexity~\cite{DBLP:conf/fsttcs/BringmannC18,BringmannK18,AbboudEtAl2014,AbboudEtAl2015,AbboudRubinstein2018}. Nevertheless, subsequences appear also in more applied settings: for modelling concurrency~\cite{Riddle1979a,Shaw1978,BussSoltys2014}, in database theory (especially \emph{event stream processing}~\cite{ArtikisEtAl2017,GiatrakosEtAl2020,ZhangEtAl2014}), in data mining~\cite{LiW08,LiYWL12}, or in problems related to bioinformatics~\cite{BilleEtAl2012}. Interestingly, a new setting, motivated by database theory \cite{Kleest-Meissner22,Kleest-Meissner23,FrochauxK23}, considers subsequences of strings, where the substrings occurring between the positions where the letters of the subsequence are embedded are constrained by regular or length constraints; a series of algorithmic results (both for upper and lower bounds) on matching and analysis problems for the sets of such subsequences occurring in 
a string were obtained \cite{DayKMS22,KoscheKMP22,AdamsonKKMS23,ManeaRS24}. 

The focus of this paper is the study of the subsequences of strings of a formal language, the main idea behind it being to extend the fundamental problems related to matching subsequences in a string and to the analysis of the sets of subsequences of a single string to the case of sets of strings. To this end, grammars (or automata) are succinct representations of (finite or infinite) sets of strings they generate (respectively, accept), so we are interested in matching and analysis problems related to the set of subsequences of the strings of a language, given by the grammar generating it (respectively, the automaton accepting it). This research direction is, clearly, not new. To begin with, we recall the famous result of Higman~\cite{higman1952ordering} which states that the downward closure of every language (i.e., the set of all subsequences of the strings of the respective language) is regular. Clearly, it is not always possible to compute an automaton accepting the downward closure of a given language, but gaining a better understanding when it is computable is an important area of research, as the set of subsequences of a language plays meaningful roles in practical applications (e.g., abstractions of complex systems, see~\cite{Zetzsche16,Zetzsche18,AnandZ23} and the references therein). Computing the downward closure of a language is a general (although, often inefficient) way to solve subsequence-matching problems for languages; for instance, we can immediately check, using a finite automaton for the downward closure, if a string occurs as subsequence of a string of the respective language. However, it is often the case that more complex analysis problems regarding the subsequences occurring in the strings of a language cannot be solved efficiently (or, sometimes, at all) using the downward closure; such a problem is to check if a given string occurs as subsequence in all the strings of a language (chosen from a complex enough class, such as the class of context-free languages). 

As a direct predecessor of this paper, motivated by similar questions, \cite{bib:UniReg}~approached algorithmic matching and analysis problems related to the universality of regular languages (for short, {\REG}). More precisely, a string over $\Sigma$ is called $k$-universal if its set of subsequences includes all strings of length $k$ over $\Sigma$; the study of these universal strings was the focus of many recent works~\cite{dlt2020,day2021edit,SchnoebelenV23} and the motivation for studying  universality properties in the context of subsequences is discussed in detail in, e.g., \cite{day2021edit,bib:UniReg}. The main problems addressed in~\cite{bib:UniReg} are the following: for $L\in{\REG}$, over the alphabet $\Sigma$, and a number $k$, decide if there exists a $k$-universal string in $L$ (respectively, if all strings of $L$ are $k$-universal). The authors of~\cite{bib:UniReg} discussed efficient algorithms solving these problems and complexity lower bounds. In this paper, we extend the work of~\cite{bib:UniReg} firstly by proposing a more structured approach for the algorithmic study of the subsequences occurring in strings of formal languages and secondly by considering more general classes of languages, both from the Chomsky hierarchy (such as the class of context-free languages or that of context-sensitive languages) and non-classical (the class of languages accepted by deterministic finite automata with translucent letters).

Our work on subsequence-matching and analysis problems in languages defined by context-free grammars (for short, {\CFG}) also extends a series of results related to matching subsequences in strings given as a straight line program (for short, {\SLP}; a {\CFG} generating a single string), or checking whether a string given as an {\SLP} is $k$-universal, for some given $k$, e.g., see~\cite{Lohrey12,SchnoebelenV23}. In our paper, we consider the case when the input context-free languages and the {\CFG}s generating them are unrestricted. 

\subparagraph*{The approached problems and an overview of our results:}\label{sec:Problems} As mentioned above, we propose a more structured approach for matching- and analysis-problems related to subsequences of the strings of a formal language. More precisely, we define and investigate the following five problems.

\begin{problem}[$\exists$-Subsequence]\label{prob:wtlsubseq}
Given a language $L$ by a machine (grammar) $M$ accepting (respectively, generating) it and a string $w$, is there a string $v\in L$ such that $w\leq v$?      
\end{problem}

\begin{problem}[$\forall$-Subsequence]\label{prob:wtl_all_subseq}
Given a language $L$ by a machine (grammar) $M$ accepting (respectively, generating) it and a string $w$, do we have for all strings $v\in L$ that $w\leq v$? 
\end{problem}

\begin{problem}[$\exists$-$k$-universal]\label{prob:exist_universal_largerthan_k}
Given a language $L$ by a machine (grammar) $M$ accepting (respectively, generating) it and integer $k$, check if there is a $k$-universal string in $ L$.
\end{problem}

\begin{problem}[$\forall$-$k$-universal]\label{prob:exist_nonuniversal}
Given a language $L$ by a machine (grammar) $M$ accepting (respectively, generating) it and integer $k$, check if all strings of $L(M)$ are $k$-universal.
\end{problem}

Alternatively, strictly from the point of view of designing an algorithmic solution, the problem above can be approached via its complement: that is, deciding if there exists at least one string in $L(M)$ which is not $k$-universal.

\begin{problem}[$\infty$-universal]\label{prob:universal_forall_m}
Given a language $L$ by a machine (grammar) $M$ accepting (respectively, generating it) decide if there exist $m$-universal strings in $ L$, for all positive integers $m$.
\end{problem}

To give some intuition on our terminology, Problems~\ref{prob:wtlsubseq} and~\ref{prob:exist_universal_largerthan_k} can be seen as {\em matching problems} (find a string which contains a certain subsequence or set of subsequences), while the other three problems are {\em analysis problems} (decide properties concerning multiple strings of the language). 

Going a bit more into details, in the main part of this paper, we investigate these problems for the case when the language $L$ is chosen from the class of context-free languages (for short, {\CFL}; given by {\CFG}s in Chomsky normal form), or from the class of context-sensitive languages (for short, {\CSL}; given by context-sensitive grammars), or from the class of languages accepted by 
deterministic finite automata with translucent letters
(given by an automaton of the respective kind). The choice of presentation of the languages from given classes, unsurprisingly, makes a big difference w.r.t.\ hardness. For instance, certain singleton languages can be encoded by {\SLP}s (essentially {\CFG}s) exponentially more succinctly than by classical {\dfa}, which of course introduces significant extra computation into solving subsequence-related queries~\cite{Lohrey12,SchnoebelenV23}. But, before approaching these classes of languages, we provide a series of general decidability results on these five problems, for which the choice of grammar or automaton as the way of specifying the input language $L$ is not consequential. 

For short, our results are the following.
We first give (in Section~\ref{sec:meta_theorems}) a series of simple sufficient conditions on a class ${\mathcal C}$ of languages (related to the computation of downward closures as well as to decidability properties for the respective class) which immediately lead to decision procedures  for the considered problems; however, these procedures are inherently inefficient, even for classes such as {\CFL}. In this context, generalizing the work of~\cite{bib:UniReg}, we approach (in subsequent sections of this paper) each of the above problems for ${\mathcal C}$ being the class {\CFL} and, respectively, the class {\CSL}. While all the problems are undecidable for {\CSL}s, we present efficient algorithms for the case of {\CFL}s. In particular, the results obtained for {\CFL} are similar to the corresponding results obtained for {\REG} (i.e., if a problem was solvable in polynomial or {\FPT}-time for {\REG}, we obtain an algorithm from the same class for {\CFL}). In that regard, it seemed natural to search for a class of languages which does not exhibit this behaviour, while retaining the decidability of (at least some of) these problems. To this end, we identify the class of languages accepted by deterministic finite automata with translucent letters (a class of automata which does not process the input in a sequential fashion) and show (in the final section of this paper) a series of initial promising results related to them.

 \section{Preliminaries} \label{sec:prels}

Let  $\N = \{1,2,\ldots\}$ denote the natural numbers and set $\N_0 = \N
\cup \{0\}$ as well as $[n]=\{1,\ldots,n\}$ and $[i,n]=\{i, i+1, \ldots, n\}$, for all $i,n\in\N_0$ with $i \leq n$.

An \emph{alphabet} $\Sigma=\{1,2,\ldots,\sigma\}$ is a finite set of symbols, called
\emph{letters}. A \emph{string} $w$ is a finite concatenation of letters from a given alphabet with the number of these letters giving its \emph{length} $|w|$. The string with no letters is the \emph{empty string} $\varepsilon$ of length $0$. The set of all finite strings over the alphabet $\Sigma$, denoted by $\Sigma^*$, is the free monoid generated by $\Sigma$ with concatenation as operation. A subset $L\in \Sigma^*$ is called a \emph{(formal) language}. Let $\Sigma^n$ denote all strings in $\Sigma^*$ exactly of length $n\in\N_0$.

For $1\leq i\leq j\leq|w|$ denote the $i^\text{th}$ letter of $w$ by $w[i]$ and the \emph{factor} of $w$ starting at position $i$ and ending at position $j$ as $w[i,j]=w[i]\cdots w[j]$. If $i=1$ the factor is also called a prefix, while if $j=|w|$ it is called a suffix of $w$. 
For each $a \in \Sigma$ set $|w|_{a} = |\{i \in [|w|] \mid w[i]=a \}|$.

Let $\alph(w)$ denote the set of all letters of $\Sigma$ occurring in $w$. 
A length $n$ string $u \in \Sigma^\ast$ is called \emph{subsequence} of $w$, denoted $u \leq w$, if $w = w_1 u[1] w_2 u[2] \cdots w_n u[n] w_{n+1}$, for some $w_1, \ldots, w_{n+1} \in \Sigma^*$. 
For $k\in\N_0$, a string $w \in \Sigma^*$ is called \emph{$k$-universal} (w.r.t.\ $\Sigma$) if	every $u\in\Sigma^k$ is a subsequence of $w$. The \emph{universality-index} $\iota(w)$ is the largest $k$ such that $w$ is $k$-universal. 

\begin{definition}\label{def:archfak}
The arch factorization of a string $w \in \Sigma^*$ is given by $w=ar_1(w)\cdots ar_{\iota(w)}(w)r(w)$ with $\iota(ar_i(w))=1$ and $ar_i(w)[|ar_i(w)|]$ $\notin \alph(ar_i(w)[1, |ar_i(w)|-1])$, for all $i \in [1,\iota(w)]$. Furthermore, $\alph(r(w)) \subsetneq \Sigma$ applies. The strings $ar_i(w)$ are called {\em arches} and $r(w)$ is called {\em the rest} of $w$. 
The modus $m(w)$ of $w$ is defined as the concatenation of the last letters of each arch:
$m(w)=ar_1(w)[|ar_1(w)|]\cdots ar_{\iota(w)}(w)[|ar_{\iota(w)}(w)|]$.
\end{definition}
As an example, in the arch factorisation $w=(bca) {\cdot} (accab) {\cdot} (cab) {\cdot} b$ of $w\in\{a,b,c\}^*$,  the parentheses denote the three arches and the rest $r(w)=b$.
Further, we have $\iota(w)=3$ and $m(bcaaccabcabb) = abb$. For more details about the arch factorization and the universality index see~\cite{TCS::Hebrard1991,dlt2020}.

A string $v$ is an absent subsequence of another string $w$ if $v$ is not a subsequence of $w$~\cite{KoscheKMS21,KoscheKMS22}. A shortest absent subsequence  of a string $w$ (for short, $\SAS(w)$) is an absent subsequence of $w$ of minimal length, i.e., all subsequences of shorter length are found in $w$. We note that, for a given word $w$ and some letter $a\notin \alph(r(w))$, an \SAS\ of $w$ is $m(w)a$~\cite{TCS::Hebrard1991,KoscheKMS21}. An immediate observation is that any string $v$ which is an $\SAS(w)$ satisfies $\len v = \iota(w) + 1$.

In this paper, we work with absent subsequences of a word $w$ which are the shortest among all absent subsequences of $w$ and, additionally, start with $a$ and end with $b$, for some $a\in \Sigma\cup\{\varepsilon\}$ and $b\in \Sigma$. Such a shortest string which starts with $a$ and ends with $b$ and is an absent subsequence of $w$ is denoted $\SAS_{a,b}(w)$. For instance, an $\SAS_{\varepsilon,b}(w)$, for some $b\notin \alph(r(w))$, is an $\SAS(w)$, such that its starting letter is not fixed, but the ending one must be $b$.

\begin{definition}
A grammar over an alphabet $\Sigma$ is a $4$-tuple $G = (V,\Sigma, P, S)$ consisting of:
a set $V = \{A,B,C, \dots\}$ of non-terminal symbols, a set $\Sigma = \{a,b,c, \dots\}$ of terminal symbols with $V \cap \Sigma = \emptyset$, a non-empty set $P \subseteq (V\cup \Sigma)^+ \times (V\cup \Sigma)^*$ of productions and a start symbol $S\in V$.
\end{definition}
We represent productions $(p,q)\in P$ by $p \rightarrow q$. In $G$, $u = xpz$ with $x,z \in (V \cup \Sigma)^*$ is directly derivable to $v=xqz$ if a production $(p,q)\in P$ exists; in this case, we write $u \Rightarrow_G v$; the subscript $G$ is omitted when no confusion arises. More generally, for $m\in\N$, we say that $u$ is derivable to $v$ in $m$ steps (denoted $w \Rightarrow_G^m v$) if there exist strings $w_0, w_1, \dots, w_m \in (V \cup \Sigma)^*$ (called sentential forms) with 
$u = w_0 \Rightarrow_G w_1 \land w_1 \Rightarrow_G w_2 \land \dots \land w_{m-1} \Rightarrow_G w_m = v.$
If $u$ is derivable to $v$ in $m$ steps, for some $m\in\N_0$, we write $u \Rightarrow_G^* v$, i.e., $\Rightarrow_G^*$ is the reflexive and transitive closure of $\Rightarrow_G$. With $L(G) = \{w \in \Sigma^* \mid S \Rightarrow_G^* w\}$ we denote the language generated by $G$. We call a derivation a sequence $S \Rightarrow \dots \Rightarrow w \in L(G)$. The number of steps used in the derivation is the derivation's length. 

\begin{definition}
A grammar $G = (V,\Sigma,P,S)$ with $P \subseteq V \times (V \cup \Sigma)^+$ is a \emph{context-free} grammar (for short, {\CFG}). A language $L$ is a context-free language (for short, \CFL) if and only if there is a {\CFG} $G$ with $L(G) = L$.
A grammar $G = (V,\Sigma,P,S)$, where for all $(p,q) \in P$ we have $|p| \leq |q|$, is a \emph{context-sensitive} grammar (for short, {\CSG}). A language $L$ is a context-sensitive language (for short, \CSL) if and only if there is a {\CSG} $G$ with $L(G) = L$.
\end{definition}

The definitions above tacitly assume that {\CFL}s and {\CSL}s do not contain the empty string $\varepsilon$. Indeed, for the problems considered here, we can make this assumption. Whether $\varepsilon \in L$ or not plays no role in deciding Problems~\ref{prob:wtlsubseq},~\ref{prob:exist_universal_largerthan_k}, and~\ref{prob:universal_forall_m}, while $\varepsilon \in L$ immediately leads to a negative answer for \cref{prob:wtl_all_subseq} (unless $w=\varepsilon$) and~\cref{prob:exist_nonuniversal} (unless $k=0$). So, for simplicity, we only address languages that, by definition, \emph{do not} contain the empty string (see also the discussions in \cite{bib:ChNF,HopcroftU79} about how the presence of $\varepsilon$ in formal languages can be handled).

Also, note that every unary {\CFL} is regular \cite{Parikh1966}, so when discussing our problems for the class of {\CFL}s we assume that the input languages are over an alphabet with at least two letters. 

\begin{definition}
A {\CFG} $G = (V,\Sigma,P,S)$ is in Chomsky normal form ({\CNF}) if and only if $P \subseteq V \times (V^2 \cup \Sigma)$ and, for all $A\in V$, there exist some $w_A, w'_A, w''_A\in \Sigma^*$ such that $A\Rightarrow^* w_A$ and $S\Rightarrow^* w'_A Aw''_A$ (these last two properties essentially say that every non-terminal of $G$ is \emph{useful}).
\end{definition}

When we discuss our problems in the case of {\CFL}s, we assume our input is a {\CFG} $G$ in {\CNF}. This does not change our results since, according to~\cite{bib:ChNF} and the references therein, we can transform any grammar $G$ in polynomial time into a {\CFG} $G^\prime$ in {\CNF} such that $|G^\prime| \in \mathcal{O}(|G|^2)$ and $L(G)=L(G')$, where $|G|$ refers to the size of a grammar determined in terms of total size of its productions. 

In some cases it may be easier to view derivations in a {\CFG} $G$ as a derivation (parse) tree. These are rooted, ordered trees. The inner nodes of such trees are labeled with non-terminals and the leaf-nodes are labeled with symbols $X \in (V \cup \Sigma)$. An inner node $A$  has, from left to right, the children $X_1, \dots, X_k$, for some integer $k\geq 1$, if the grammar contains the production $A \rightarrow X_1\dots X_k$. As such, if we concatenate, from left to right, the leaves of a derivation tree $T$ with root $A$ we get a string $\alpha$ (called in the following the border of $T$) such that $A\rightarrow^* \alpha$. The depth of a derivation tree is the length of the longest simple-path starting with the root and ending with a leaf (i.e., the number of edges on this path). If $G$ is in {\CNF}, then all its derivation trees are binary.

\begin{definition}
    For any language $L\subset\Sigma^\ast$ the {\em downward closure} $\Down L$ of $L$ is defined as the language containing all subsequences of strings of $L$, i.e., $\Down L =\{v\in\Sigma^\ast\mid \exists w\in L:v\leq w\}$. The complementary notion of the {\em upward closure} $\Up L$ of a language $L$ is the language containing all supersequences of strings in $L$, i.e., $\Up L =\{w\in\Sigma^\ast\mid \exists v\in L: v\leq w\}$. 
\end{definition}

Our problems focus on properties of formal languages, and Problems~\ref{prob:exist_universal_largerthan_k},~\ref{prob:exist_nonuniversal} and~\ref{prob:universal_forall_m} are strongly connected to universality seen as a property of a language, therefore we extend the concept of universality to formal languages. We distinguish between two different ways of analyzing the universality of a language.
\begin{definition}
Let $L\subseteq \Sigma^*$ be a language. We call $L$ $k$-universal universal if for every $w \in L$ it holds that $\iota(w) \geq k$. The universal universality index $\iotaA(L)$ is the largest $k$, such that $L$ is $k$-universal universal. We call $L$ $k$-existential universal if a string $w \in L$ exists with $\iota(w) \geq k$. The existential universality index $\iotaE(L)$ is the largest $k$, such that $L$ is $k$-existential universal. In all the definitions above, the universality index of words and, respectively, languages is computed w.r.t. $\Sigma$. 
\end{definition}

In case of a singleton language $L=\{w\}$ it holds that $\iotaE(L)=\iotaA(L)=\iota(w)$. In general the universal universality index $\iotaA(L)$ is the infimum of the set of all universality indices of strings in $L$ and therefore is lower bounded by $0$ and upper bounded by $\iota(w)$, for any $w\in L$ (so it is finite, for $L\neq \emptyset$). The existential universality index $\iotaE(L)$ is the supremum of the set of all universality indices of strings in $L$ and, as such, can be infinite. In this setting the answer to \cref{prob:exist_universal_largerthan_k} and, respectively, \cref{prob:exist_nonuniversal}, can be solved by computing $\iotaE(L)$ and, respectively, $\iotaA(L)$, and then checking whether $k\leq \iotaE(L)$ and, respectively, $k\leq\iotaA(L)$. Furthermore, \cref{prob:universal_forall_m} asks whether $\iotaE(L)$ is infinite or not. 
The following two lemmas are not hard to show. 
\begin{lemma}\label{lem:dfa_superseq}
    Given a string $w\in \Sigma^*$, with $|w|=n$ and $|\Sigma|=\sigma$, we can construct in time $\bigo(n\sigma )$ a minimal {\dfa}, with $n+1$ states, accepting the set of strings which have $w$ as a subsequence.
\end{lemma}
\begin{proof}
    Let $n=|w|$. We construct a {\dfa} $A=(Q,\Sigma, q_0,\{f\},\delta)$ accepting the set of strings which have $w$ as a subsequence. We start by defining the set of states $Q=\{0,1,\ldots,n\}$, the initial state $q_0=0$, and the single final state $f=n$. Now, for the transition function, we set, for $i\in [n]$, $\delta(i-1,a)=i$ if and only if $w[i]=a$; otherwise, we set $\delta(i,a)=i$, for $i\in Q$ and $a\in \Sigma$. 
    
    The correctness of the construction (i.e., the fact that $A$ accepts exactly the set of strings $v$ with $w\leq v$) is immediate. Indeed, to go from state $0$ to state $n$ we need to read a string which contains, in order but not necesarilly on consecutive positions, the letters $w[1], w[2], \ldots, w[n]$. That is, we need to read a string which has $w$ as a subsequence.

    To see that this {\dfa} is minimal, it is enough to observe that the strings $w[1:i]$, for $i\geq 0$, are not equivalent w.r.t.\ the Myhill-Nerode equivalence. So, any {\dfa} accepting the set of strings $v$ with $w\leq v$ needs to have at least $n+1$ states.

    This concludes our proof.
\end{proof}

\begin{lemma}\label{lem:dfa_kuniv}
    For $k>0$ and an alphabet $\Sigma$ with $|\Sigma|=\sigma$ we can construct in time $\bigo(2^{\sigma} k\ \poly(\sigma))$ a minimal {\dfa}, with $(2^{\sigma}-1) k +1$ states, accepting the set of $k$-universal strings over $\Sigma$. 
\end{lemma}
\begin{proof}
    Assume, for simplicity, that $\Sigma = \{1,2,\ldots, \sigma\}$. 
    
    We construct a {\dfa} $A=(Q,\Sigma, q_0,\{f\},\delta)$ accepting the set of $k$-universal strings. We start by defining the set of states $Q=\{(i,S)\mid 0\leq i\leq k-1, S\subsetneq \Sigma\} \cup \{(k)\}$, the initial state $q_0=(0,\emptyset)$, and the single final state $f=(k)$. The transition function is defined as follows:
    \begin{itemize}
        \item for $a\in \Sigma$, $\delta((k),a)=(k)$;
        \item for $0\leq t<k$, $S\subset \Sigma$, and $a\in \Sigma$, if $a\notin S$ and $S\cup \{a\}\neq \Sigma$, then $\delta((t,S), a)= (t,S\cup\{a\})$;
        \item for $0\leq t<k-1$, $S\subset \Sigma$, and $a\in \Sigma$, if $a\notin S$ and $S\cup \{a\}= \Sigma$, then $\delta((t,S), a)= (t+1,\emptyset)$;
        \item for $S\subset \Sigma$, and $a\in \Sigma$, if $a\notin S$ and $S\cup \{a\}= \Sigma$, then $\delta((k-1,S), a)= (k)$;
        \item for $0\leq t<k$, $S\subset \Sigma$, and $a\in S$, then $\delta((t,S), a)= (t,S)$.
    \end{itemize}

    To show that $A$ accepts the set of $k$-universal strings, we make the following observations:
    \begin{itemize}
        \item The set of strings which label paths starting with $(0,\emptyset)$ and ending with $(0,S)$, for some set $S\subsetneq \Sigma$, are exactly those strings $w$, with $\alph(w)=S$. This can be shown easily by induction on $|S|$.
        \item The set of strings which label paths starting with $(0,\emptyset)$ and ending with $(1,\emptyset)$ are exactly the $1$-universal strings $w$, such that $r(w)=\varepsilon$. This follows immediately from the previous observation. 
        \item For $1\leq i\leq k-2$, the set of strings which label paths starting with $(0,\emptyset)$ and ending with $(i,S)$, for some set $S\subsetneq \Sigma$ with $S\neq \emptyset $, are exactly those strings $w$, with $\iota(w)=i$ and $\alph(r(w))=S $, and the set of strings which label paths starting with $(0,\emptyset)$ and ending with $(i+1,\emptyset)$ are exactly those strings $w$, with $\iota(w)=i+1$ and $r(w)=\varepsilon$. This can be shown by induction. We first  assume that the property holds for $j= i-1$ (and note that it holds for $j=0$). Then we show, by induction on $|S|$, that the set of strings which label paths starting with $(0,\emptyset)$ and ending with $(i,S)$, for some set $S\subsetneq \Sigma$, are exactly those strings $w$, with $\iota(w)=i$ and $\alph(r(w))=S $. Finally, it immediately follows that the set of strings which label paths starting with $(0,\emptyset)$ and ending with $(i+1,\emptyset)$ are exactly those strings $w$, with $\iota(w)=i+1$ and $r(w)=\varepsilon$. 
        \item The set of strings which label paths starting with $(0,\emptyset)$ and ending with $(k-1,S)$, for some set $S\subsetneq \Sigma$ with $S\neq \emptyset $, are exactly those strings $w$, with $\iota(w)=k-1$ and $\alph(r(w))=S $, and the set of strings which label paths starting with $(0,\emptyset)$ and ending with $(k)$ are exactly those strings $w$, with $\iota(w)\geq k$. Again, we can show by induction on $|S|$, that the set of strings which label paths starting with $(0,\emptyset)$ and ending with $(k-1,S)$, for some set $S\subsetneq \Sigma$, are exactly those strings $w$, with $\iota(w)=i$ and $\alph(r(w))=S $. Finally, to reach $(k)$ directly from a state $(k-1,S)$ we need to have that $\Sigma\setminus S=\{a\}$ and we read $a$ in state $(k-1,S)$; this means that the string we have read is $k$-universal. Once $(k)$ is reached, it is never left no matter what we read, so the strings that lead from $(0,\emptyset)$ to $(k)$ are exactly the $k$-universal strings (so with $\iota(w)\geq k$). 
    \end{itemize}

    Further, we show that $A$ is minimal. Consider $0\leq p\leq \sigma$ and $0\leq i\leq k-1$, and a strict subset $S=\{j_1,\ldots,j_p\}$ of $\Sigma$. Now, consider the strings $w_{i,S}=(12\cdots \sigma)^i (j_1 \cdots j_p)$. It is immediate that $w_{i,S}$ and $w_{j,S'}$ are not equivalent under the Myhill-Nerode equivalence (defined for the language of $k$-universal strings over $\Sigma$) if $(i,S)\neq (j,S')$.  Moreover, none of these strings is equivalent to $(12\cdots \sigma)^k$. So, the Myhill-Nerode equivalence has at least $(2^\sigma -1)k +1 $ states. As $A$ has exactly as many states, our statement follows.    
\end{proof}

The computational model we use to state our algorithms is the standard unit-cost word RAM with logarithmic word-size $\omega$ (meaning that each memory word can hold $\omega$ bits). It is assumed that this model allows processing inputs of size $n$, where $\omega \geq \log n$; in other words, the size $n$ of the data never exceeds (but, in the worst case, is equal to) $2^\omega$. Intuitively, the size of the memory word is determined by the processor, and larger inputs require a stronger processor (which can, of course, deal with much smaller inputs as well). Indirect addressing and basic arithmetical operations on such memory words are assumed to work in constant time. Note that numbers with $\ell$ bits are represented in $O(\ell/\omega )$ memory words, and working with them takes time proportional to the number of memory words on which they are represented. This is a standard computational model for the analysis of algorithms, defined in \cite{FredmanW90}. 

Our algorithms have languages as input, that is sets of strings over some finite alphabet. Therefore, we follow a standard stringology-assumption, namely that we work with an {\em integer alphabet}: we assume that this alphabet is $\Sigma=\{1,2,\ldots,\sigma\}$, with $|\Sigma|=\sigma$, such that $\sigma$ fits in one memory word. For a more detailed general discussion on the integer alphabet model see, e.\,g.,~\cite{crochemore}. In all problems discussed here, the input language is given as a grammar generating it or as an automaton accepting it. We assume that all the sets defining these generating/accepting devices (e.g., set of non-terminals, set of states, set of final states, relation defining the transition function or derivation, etc.) have at most~$2^\omega$ elements and their elements are integers smaller or equal to $2^\omega$ (i.e., their cardinality and elements can be represented as integers fitting in one memory word). In some of the problems discussed in this paper, we assume that we are given a number $k$. Again, we assume that this integer fits in one memory word.

One of our algorithms (for \cref{prob:exist_universal_largerthan_k} in the case of $\CFL$, stated in Theorem \ref{thm:prob_exist_universal_largerthan_k}) runs in exponential time and uses exponential space w.r.t.\ the size of the input. In particular, both the space and time complexities of the respective algorithm are exponential, with constant base, in $\sigma$ (the size of the input alphabet) but polynomial w.r.t.\ all the other components of the input. To avoid clutter, we assume that our exponential-time and -space algorithm runs on a RAM model where we can allocate as much memory as our algorithms needs (i.e., the size of the memory-word $\omega$ is big enough to allow addressing all the memory we need in this algorithm in constant time); for the case of $\sigma \in O(1)$, this additional assumption becomes superfluous, and for non-constant $\sigma$ we stress out that the big size of memory words is only used for building large data structures, but not for speeding up our algorithms by, e.g., allowing constant-time operations on big numbers (that is, numbers represented on more than $c\log N$ bits, for some constant $c$ and $N$ being the size of the input).

 \section{General Results}\label{sec:meta_theorems}

We consider the problems introduced in \cref{sec:Problems}, for the case when the language $L$ is chosen from a class ${\mathcal C}$, 
and give a series of sufficient conditions for them to be decidable. 

Consider a class ${\mathcal G}$ of grammars (respectively, a class ${\mathcal A}$ of automata) generating (respectively, accepting) the languages of the class ${\mathcal C}$. For simplicity, for the rest of this section, we assume that in all the problems we take as input a grammar $G_L$ such that $L(G_L)=L$, but note that all the results hold for the case when we consider that the languages are given by an automaton from the class ${\mathcal A}$ accepting them. 

Let ${\mathcal C}'$ be the class of languages $L\cap R$, where $L\in {\mathcal C}$ and $R\in{\REG}$. We use two hypotheses:
\begin{itemize}
    \item[H1.] Given a grammar $G$ of the class ${\mathcal G}$ we can algorithmically construct a non-deterministic finite automaton $A$ accepting the downward closure of $L(G)$. 
    \item[H2.] Given a grammar $G$ of the class ${\mathcal G}$ and a non-deterministic finite automaton $A$, we can algorithmically decide whether the language $L(G)\cap L(A)$ is empty. 
\end{itemize}

First 
we show that, under H1, Problems~\ref{prob:wtlsubseq},~\ref{prob:exist_universal_largerthan_k}, and~\ref{prob:universal_forall_m} are decidable. 
\begin{theorem}\label{thm:DC_dec}
If H1 holds, then Problems~\ref{prob:wtlsubseq},~\ref{prob:exist_universal_largerthan_k}, and~\ref{prob:universal_forall_m} are decidable. 
\end{theorem}
\begin{proof}
We start by observing that the following straightforward properties hold:
\begin{itemize}
    \item for a string $w$, there exists $v\in L$ such that $w\leq v$ if and only if $w \in \Down L$.
    \item for an integer $k>0$, there exists $v\in L$ such that $v$ is $k$-universal if and only if there exists $v' \in \Down L$ such that $v'$ is $k$-universal.
\end{itemize}
In each of Problems~\ref{prob:wtlsubseq},~\ref{prob:exist_universal_largerthan_k}, and~\ref{prob:universal_forall_m}, we are given a grammar $G$ generating the language $L$. According to H1, we construct a non-deterministic automaton $A$ accepting $\Down L$, the downward closure of $L$. 

For Problem~\ref{prob:wtlsubseq}, it is sufficient to check if $L(A)=\Down L$ contains the string $w$, which is clearly decidable. For Problem~\ref{prob:exist_universal_largerthan_k} we need to decide if $L$ contains a $k$-universal string. By our observations, it is enough to check if $\Down L$ contains a $k$-universal string. This can be decided, for the automaton $A$, according to~\cite{bib:UniReg}. For Problem~\ref{prob:universal_forall_m} we need to decide if $L$ contains a $k$-universal string, for all $k\leq 0$. This is also decidable, for $A$, according to the results of~\cite{bib:UniReg}. 
\end{proof}

Secondly, we show that, under H2, Problems~\ref{prob:wtlsubseq},~\ref{prob:wtl_all_subseq},~\ref{prob:exist_universal_largerthan_k}, and~\ref{prob:exist_nonuniversal} are decidable. 
\begin{theorem}\label{thm:intersection_dec}
If H2 holds, then Problems~\ref{prob:wtlsubseq},~\ref{prob:wtl_all_subseq},~\ref{prob:exist_universal_largerthan_k}, and~\ref{prob:exist_nonuniversal} are decidable. 
\end{theorem}
\begin{proof}
In all the inputs of Problems~\ref{prob:wtlsubseq},~\ref{prob:wtl_all_subseq},~\ref{prob:exist_universal_largerthan_k}, and~\ref{prob:exist_nonuniversal} when considering $CFL$ and $CSL$, we are given a grammar $G$, which generates the language $L$. 

For Problems~\ref{prob:wtlsubseq} and~\ref{prob:wtl_all_subseq}, by Lemma~\ref{lem:dfa_superseq} we construct a {\dfa} $B$ accepting the regular language $\Up w$ of strings which have $w$ as a subsequence. If the intersection of $L$ (given as the grammar $G$ which generates it) and $L(B)$ is empty, which is decidable, under H2, then the considered instance of Problem~\ref{prob:wtlsubseq} is answered negatively; otherwise, it is answered positively. By making the final state of $B$ non-final, and all the other states final, we obtain a {\dfa} $B'$ which accepts $\Sigma^*\setminus\Up w$. If the intersection of $L$ and $L(B')$ is empty, then the answer to the considered instance of Problem~\ref{prob:wtl_all_subseq} is positive; otherwise, it is negative. 

For Problems~\ref{prob:exist_universal_largerthan_k} and~\ref{prob:exist_nonuniversal}, by Lemma~\ref{lem:dfa_kuniv} we construct a {\dfa} $B$ accepting the regular language of $k$-universal strings. If the intersection of $L$ and $L(B)$ is empty, then the answer to the considered instance of Problem~\ref{prob:exist_universal_largerthan_k} is negative; otherwise, it is positive. By making the final state of $B$ non-final, and all the other states final, we obtain a {\dfa} $B'$ which accepts exactly all strings which are not $k$-universal. If the intersection of $L$ and $L(B')$ is empty, then the answer to the considered instance of Problem~\ref{prob:exist_nonuniversal} is positive; otherwise, it is negative. 
\end{proof}

It is worth noting that, even for classes which fulfill both hypotheses above (such as the {\CFL}s \cite{Leeuwen78,HopcroftU79}), there are several reasons why the algorithms resulting from the above theorems are not efficient. On the one hand, constructing an automaton which accepts the downward closure of a language is not always possible, and even when this construction is possible (when the language is from a class for which H1 holds) it cannot always be done efficiently. For instance, in the case of {\CFL}s, this may take inherently exponential time w.r.t.\ the size of the input grammar~\cite{GruberHK09}; in this paper, we present more efficient algorithms for Problems~\ref{prob:wtlsubseq},~\ref{prob:wtl_all_subseq},~\ref{prob:exist_universal_largerthan_k}, and~\ref{prob:exist_nonuniversal} in the case of {\CFL}s, which do not rely on \cref{thm:intersection_dec}. On the other hand, the results of \cref{thm:intersection_dec} rely, at least partly, on the construction of a {\dfa} accepting all $k$-universal strings, which takes exponential time in the worst case, as it may have an exponential number of states (both w.r.t.\ the size of the input alphabet and w.r.t.\ the binary representation of the number $k$, which is given as input for some of these problems). 

Interestingly, the class {\CSL} does not fulfil any of the above hypotheses. In fact, as our last general result, we show that all five problems we consider here are undecidable for {\CSL}. 
\begin{theorem}\label{thm:undecCSL}
Problems~\ref{prob:wtlsubseq},~\ref{prob:wtl_all_subseq},~\ref{prob:exist_universal_largerthan_k},~\ref{prob:exist_nonuniversal},~\ref{prob:universal_forall_m} are undecidable for the class of {\CSL}, given as {\CSG}s.
\end{theorem}
\begin{proof}
To obtain the undecidability of all the problems, we show reductions from the emptiness problem for Context Sensitive Languages. Assume that we have a {\CSL} $L$, specified by a grammar $G$, as the input for the emptiness problem for {\CSL}. Assume $L$ is over the alphabet $\Sigma = \{ 1, \dots, \sigma\}$, and that the {\CSG} $G$, has the starting symbol $S$. Let $0$ be a fresh letter (i.e., $0\notin \Sigma$). 
 
To show the undecidability of Problems~\ref{prob:wtlsubseq} and~\ref{prob:wtl_all_subseq}, we construct a new grammar $G'$ which has all the non-terminals, terminals, and productions of $G$ and, additionally, $G'$ has a new starting symbol $S'$ and the productions $S'\rightarrow \sigma S$ and $S'\rightarrow 0$. 

It is immediate that there exists a string $w\in L(G')$ which contains $\sigma$ as a subsequence, if and only if $L(G)$ is not empty. Furthermore, all strings of $L(G')$ contain $0$ as a subsequence (that is, the production $S'\rightarrow \sigma S$  is not the first production in the derivation of any terminal string) if and only if $L(G)$ is empty. As the emptyness problem is undecidable for {\CSL} (given as grammars), it follows that Problems~\ref{prob:wtlsubseq} and~\ref{prob:wtl_all_subseq} are also undecidable for this class of languages.

To show the undecidability of Problem~\ref{prob:exist_universal_largerthan_k}, we construct a new grammar $G'$ which has all the non-terminals, terminals, and productions of $G$ and, additionally, $G'$ has a new starting symbol $S'$ and the production $S'\rightarrow (1 2 \cdots \sigma) S$. Clearly, $L(G')$ contains a $1$-universal string (over $\Sigma$) if and only if $L(G)\neq \emptyset$. Thus, it follows that Problem~\ref{prob:exist_universal_largerthan_k} is also undecidable for this class of languages.

To show the undecidability of Problem~\ref{prob:exist_nonuniversal}, we construct a new grammar $G'$ which has all the non-terminals, terminals, and productions of $G$ and, additionally, $G'$ has a new starting symbol $S'$ and the productions $S'\rightarrow 0 1 2\cdots \sigma $ and $S'\rightarrow S$. Clearly, all the strings of $L(G')$ are $1$-universal (over $\Sigma\cup \{0\}$) if and only if $L(G)= \emptyset$ (as any string which would be derived in $G'$ starting with the production $S'\rightarrow S$ would not contain $0$). Hence, Problem~\ref{prob:exist_nonuniversal} is also undecidable for {\CSL}. 

To show the undecidability of Problem~\ref{prob:universal_forall_m}, we construct a new grammar $G'$ which has all the non-terminals, terminals, and productions of $G$ and, additionally, $G'$ has a new starting symbol $S'$ and a fresh non-terminal $R$ and the productions $S'\rightarrow 0 1 2\cdots \sigma $, $S'\rightarrow R S$, $R\rightarrow 0 1 \cdots \sigma R$, and $R\rightarrow 0 1 \cdots \sigma$. Clearly, $L(G')$ contains $m$-universal strings (over $\Sigma\cup \{0\}$), for all $m\geq 1$, if and only if $L(G)\neq \emptyset$ (as we can use $R$ to pump arches in the strings of $L(G')$ if and only if there exists at least one derivation where $S$ can be derived to a terminal string). Accordingly, Problem~\ref{prob:universal_forall_m} is also undecidable for {\CSL}. 
\end{proof}

Given that all the problems become undecidable for ${\mathcal C}={\CSL}$, we now focus our investigation on classes of languages strictly contained in the class of {\CSL}s.

 \section{Problems~\ref{prob:wtlsubseq} and~\ref{prob:wtl_all_subseq} }
\label{sec:Problem1}

For the rest of this section, assume that $|w|=m$ and $|\Sigma|=\sigma$. Let us begin by noting that Problems~\ref{prob:wtlsubseq} and~\ref{prob:wtl_all_subseq} can be solved in polynomial time for the class {\REG} following the approach of Theorem~\ref{thm:intersection_dec}. Indeed, in this case, we assume that $L$ is specified by the NFA $A$, with $s$ states, with $L(A)=L$, and then we either have to check the emptiness of the intersection of $L=L(A)$ with the language accepted by the deterministic automaton constructed in Lemma~\ref{lem:dfa_superseq}, or, respectively, with the complement of this language; both these tasks clearly take polynomial time. 

We now consider the two problems for the class of {\CFL}s. We first recall the following folklore lemma (see, e.g., \cite{HopcroftU79}).
\begin{lemma}\label{lem:intestectCF_REG}
    Let $G=(V,\Sigma, P, S)$ be a {\CFG} in {\CNF}, and let $A=(Q, \Sigma, q_0, F, \delta)$ be a {\dfa}. Then we can construct in polynomial time a {\CFG} $G_A$ in {\CNF} such that $L(G_A)=L(G)\cap L(A)$.
\end{lemma}
\begin{proof}
Let $G=(V,\Sigma, P, S)$ be a {\CFG} in {\CNF}, with $|V|=n$ and $|\Sigma|=\sigma$, and let $A=(Q, \Sigma, q_0, F, \delta)$ be a {\dfa}, with $|Q|=\ell$.

We first define the {\CFG} $G'=(V',\Sigma, P', S')$ as follows:
\begin{itemize}
\item $V'=\{S'\}\cup \{(q,A,q')\mid A\in V, q,q'\in Q\}$;
\item $P'=\{S' \rightarrow (q_0,S,f)\mid f\in F\}\cup \{(q,A,q')\rightarrow (q,B,q'')(q'',C,q')\mid A\rightarrow BC\in P, q,q',q''\in Q\}\cup \{(q,A,q')\rightarrow a\mid A\rightarrow a\in P, q,q'\in Q, q'\in \delta(q,a)\}$.
\end{itemize}
It can now be easily shown by induction on the length $t$ of the derivation that $S'\Rightarrow^t_{G'} w$ if and only if $w\in L(G)\cap L(A)$. 
The grammar $G$ can be na\"ively constructed in time $\bigo(n^3\ell^3 + n\ell\sigma)$. 

Following the construction of \cite{bib:ChNF}, $G'$ can then be further transformed into a grammar $G_A$ in {\CNF} with $L(G_A)=L(G')$ in time $\bigo(n^3\ell^3)$ (in this case, only the renamings of the starting symbol $S'$ need to be eliminated).
\end{proof}

We can now state the main result of this section. We can apply Lemma \ref{lem:intestectCF_REG}, for Problem \ref{prob:wtlsubseq}, to the input {\CFG} and the {\dfa} constructed in Lemma~\ref{lem:dfa_superseq}, or, for Problem \ref{prob:wtl_all_subseq}, to the input {\CFG} and the complement of the respective {\dfa}. In both cases, we compute a {\CFG} in {\CNF} generating the intersection of a {\CFL} and a {\REG} language, and we have to check whether the language generated by that grammar is empty or not; all these can be implemented in polynomial time. 

\begin{theorem}\label{thm:Pb1_CFL}
Problems~\ref{prob:wtlsubseq} and~\ref{prob:wtl_all_subseq}, for an input grammar with $n$ non-terminals and an input string of length $m$, are decidable in polynomial time for {\CFL}.
\end{theorem}
\begin{proof}
We assume that $L$ is specified by a {\CFG} in {\CNF} $G$. Let $A$ be the {\dfa} constructed in Lemma~\ref{lem:dfa_superseq}, which accepts all strings having $w$ as a subsequence, and let $B$ be the {\dfa} obtained from $A$ by making all its non-final states final, and its only final state non-final (i.e., the complement {\dfa} for $A$, accepting $\Sigma^*\setminus L$). The construction of $A$ and $B$ takes polynomial time $\bigo(m\sigma)$, and they both have $m+1$ states.

By Lemma~\ref{lem:intestectCF_REG}, in both cases we can construct {\CFG} $G_A$ and $G_B$, respectively, generating $L\cap L(A)$, respectively $L\cap L(B)$. The complexity of this step is $\bigo(n^3m^3 + nm\sigma)$.

Further, we follow the approach of Theorem~\ref{thm:intersection_dec}. For Problem~\ref{prob:wtlsubseq}, we check if the intersection of $L$ and $L(A)$ is empty or not. This can be done in polynomial time by checking if $G_A$ generates any string (see, e.g., \cite{HopcroftU79}). If the intersection is empty, then the answer to the considered instance of Problem~\ref{prob:wtlsubseq} is answered negatively; otherwise, it is answered positively. For Problem~\ref{prob:wtl_all_subseq}, we check if the intersection of $L$ and $L(B)$ is empty or not. Again, this can be done in polynomial time by checking if $G_B$ generates any string. If the intersection is empty, then the answer to the considered instance of Problem~\ref{prob:wtl_all_subseq} is answered positively; otherwise, it is answered negatively. The time used for checking the emptiness of the respective {\CFL}s is, in both cases, $\bigo((n^3m^3 + nm\sigma)nm^2)$.
\end{proof}

 \section{Problems~\ref{prob:exist_universal_largerthan_k} and~\ref{prob:universal_forall_m}}
\label{sec:Problem3}

Let us begin by noting that in \cite{bib:UniReg} it was shown that for a given {\nfa} $A$ with $s$ states (with input alphabet $\Sigma$, where $|\Sigma|=\sigma$) and an integer $k\geq 0$, we can decide whether $L(A)$ contains a $k$-universal string (i.e., Problem~\ref{prob:exist_universal_largerthan_k} for the class {\REG}) in time $\bigo(\poly (s,\sigma)2^\sigma)$; in other words, Problem~\ref{prob:exist_universal_largerthan_k} is fixed parameter tractable ({\FPT}) w.r.t.\ the parameter $\sigma$. A polynomial time algorithm  (running in $\bigo(\poly(s,\sigma)$ time) was given for Problem~\ref{prob:universal_forall_m}, relying on the observation that, given an {\nfa} $A$, the language $L(A)$ contains strings with arbitrarily large universality if and only if $A$ contains a state $q$, which is reachable from the initial state and from which one can reach a final state, and a cycle which contains this state, whose label is $1$-universal. Coming back to Problem~\ref{prob:exist_universal_largerthan_k} for REG, the same paper shows that it is actually NP-complete. This is proved by a reduction from the Hamiltonian Path problem ({\HPP}, for short), in which a graph with $n$ vertices, the input of {\HPP}, is mapped to an input of Problem~\ref{prob:exist_universal_largerthan_k} consisting in an automaton with $\bigo(n^2)$ states over an alphabet of size $n$. This reduction also implies that, assuming that the Exponential Time Hypothesis ({\ETH}, for short) holds, there is no $2^{o(\sigma)}\poly(s,\sigma)$ time algorithm solving Problem~\ref{prob:exist_universal_largerthan_k} (as this would imply the existence of an $2^{o(n)}$ time algorithm solving {\HPP}); see \cite{LokshtanovMS11} for more details related to the {\ETH} and {\HPP}. 

Further, we consider Problems~\ref{prob:exist_universal_largerthan_k} and~\ref{prob:universal_forall_m} for the class $\CFL$, and we assume that, in both cases, we are given a {\CFL} $L$ by a 
{\CFG} $G = (V, \Sigma, P, S)$ in {\CNF}, with $n$ non-terminals, over an alphabet $\Sigma$, with $\sigma$ letters, and an integer $k\geq 1$ (in binary representation).

To transfer the lower bound derived for Problem~\ref{prob:exist_universal_largerthan_k} in the case of {\REG} (specified as {\nfa}s) to the larger class of {\CFL}s, we recall the folklore result that a {\CFG} in {\CNF} can be constructed in polynomial time from an {\nfa} (by constructing a regular grammar from the {\nfa}, and then putting the grammar in \CNF, see \cite{HopcroftU79}). So, the same reduction from \cite{bib:UniReg} can be used to show that, assuming {\ETH} holds, there is no $2^{o(\sigma)}\poly(n,\sigma)$ time algorithm solving Problem~\ref{prob:exist_universal_largerthan_k}. This reduction shows also that Problem~\ref{prob:exist_universal_largerthan_k} is NP-hard; whether this problem is in NP remains open. 

We now focus on the design of a $2^{\bigo(\sigma)}\poly(n,\sigma)$ time algorithm solving Problem~\ref{prob:exist_universal_largerthan_k} (which would also show that this problem is {\FPT}) and show that Problem~\ref{prob:universal_forall_m} can be solved in polynomial time. 
Let us recall that Problem~\ref{prob:universal_forall_m} requires deciding whether $\iotaE(L)$ is finite, and, if yes, Problem~\ref{prob:exist_universal_largerthan_k}  requires checking whether $\iotaE(L)\geq k$. 

We start by introducing a new concept which leads to a series of combinatorial observations. 

\begin{definition}
Let $G = (V,\Sigma,P,S)$ be a {\CFG}. A non-terminal $A\in V$ generates a \emph{1-universal cycle} if and only if there exists a derivation $A \Rightarrow^* w_1 A w_2$ of the grammar $G$ with $w_1,w_2\in \Sigma ^*$ and $\max(\iota(w_1),\iota(w_2)) \geq 1$.
\end{definition}
We can show the following result.

\begin{lemma}\label{lem:one_uni_cycle}
Let $G = (V, \Sigma, P, S)$ be a {\CFG} in {\CNF} and $L = L(G)$. Then $\iotaE(L)$ is infinite if and only if there exists a non-terminal $X \in V$ such that $X$ generates a $1$-universal cycle. 
\end{lemma}
\begin{proof}
Assume we have a non-terminal $A\in V$ which generates a \emph{1-universal cycle}. This means that there exists a derivation $A \Rightarrow^* w_1 A w_2$ with $w_1,w_2\in \Sigma ^*$ and $\iota(w_1) \geq 1$ or $\iota(w_2) \geq 1$. As $G$ is in {\CNF}, we have that there exist $w'_A,w''_A\in \Sigma^*$ and the derivation $S\Rightarrow^* w'_A Aw''_A$, and, also, that there exists $w_A\in \Sigma^*$ such that $A\Rightarrow^* w_A$. We immediately get that, for all $n\geq 1$, the following derivation is valid: $S\Rightarrow^* w'_A Aw''_A \Rightarrow^* w'_A w_1 A w_2 w''_A \Rightarrow^* w'_A (w_1)^2 A (w_2)^2 w''_A \Rightarrow^* w'_A (w_1)^n A (w_2)^n w''_A \Rightarrow^* w'_A (w_1)^n w_A (w_2)^n w''_A=w$. As $\iota(w_1) \geq 1$ or $\iota(w_2) \geq 1$, it follows that $\iota(w)\geq n$. So, $\iotaE(L)$ is infinite. 

We now show the converse implication. More precisely, we show by induction on the number of non-terminals of $G$ that if $\iotaE(L(G))$ is infinite then $G $ has at least one useful non-terminal $X \in V$ such that $X$ has a $1$-universal cycle. For this induction proof, we can relax the restrictions on $G$: more precisely, we still assume that the set $P$ of productions of $G$ fulfils $P \subseteq V \times (V^2 \cup \Sigma)$ but do not require that every non-terminal of $G$ is useful; it suffices to require the starting symbol to be useful.

The result is immediate if $G$ has a single non-terminal, i.e., the start symbol $S$. We now assume that our statement holds for {\CFL}s generated by grammars with at most $m$ non-terminals, and assume that $L$ is a {\CFL} generated by a {\CFG} $G$ with $m+1$ non-terminals. We want to show that $G$ has at least one useful non-terminal $X \in V$ such that $X$ has a $1$-universal cycle. We can assume, w.l.o.g., that $S$ does not have a $1$-universal cycle (otherwise, the result already holds). 

Now, consider for each useful $A\in V\setminus \{S\}$ the {\CFG} (which fulfills the requirements of our statement) $G_A=(V\setminus \{S\}, \Sigma, A, P')$, where $P'$ is obtained from $P$ by removing all productions involving $S$. Clearly, if there exists some $A\in V$ such that $\iotaE(L(G_A))$ is infinite, then, by induction, $G_A$ contains a useful non-terminal $X \in V$ such that $X$ has a $1$-universal cycle. As $G_A$ is obtained from $G$ by removing some productions and one non-terminal, it is clear that $X$ also has a $1$-universal cycle in $G$ and is also useful in $G$, so our statement holds. Let us now assume, for the sake of a contradiction, that, for each useful $A\in V$, there exists an integer $N_A\geq 1$ such that $\iotaE(L(G_A))\leq N_A$. Take $N=1+\max\{N_A\mid A\in V\}$. As $\iotaE(L)$ is infinite, there exists a string $w \in L(G)$ with $\iota(w)\geq 2N+3$. Since $w \in L(G)$, $S \Rightarrow^* w$ holds. 

Let $T_S$ be the derivation tree of $w$ with root $S$ and note that all non-terminals occurring in $T_S$ are useful. Let $p$ the longest simple path of $T_S$ starting in $S$ and having the end-node $S$ (in the case when there are more such paths, we simply choose one of them). We denote by $T^p_S$ the sub-tree of $T_S$ rooted in the end-node of $p$. If $w^\p$ is the string obtained by reading the leaves of $T^p_S$ left-to-right, then we have the following derivation corresponding to $T_S$: $S\Rightarrow v_S S v^\p_S \Rightarrow^*  v_S w^\p v^\p_S=w$, where $v_S v^\p_S\in \Sigma^*$. Since, by our assumption, $S$ does not have a $1$-universal cycle, we get that $\iota(v_S)=0$, $\iota(v^\p_S) = 0$, and that $\iota (w^\p)\geq 2N+1$.

Further, we consider $T^p_S$, and note that no other node of this tree, except the root, is labelled with $S$. Assume that the first step in the derivation $S\Rightarrow^* w^\p$ is $ S\Rightarrow AB$, for some non-terminals $A,B\in V$ and production $S\rightarrow AB$, and that the children of the root $S$ in the tree $T^p_S$ are the sub-trees $T_A$ and $T_B$. Let $w_A$ be the border of $T_A$ and $w_B$ be the border of $T_B$. Clearly, it follows that at least one of the strings $w_A$ and $w_B$ is $N$-universal. We can assume, w.l.o.g., that $\iota(w_A)\geq N$. But $w_A\in L(G_A)$ and $\iotaE(L(G_A))<N$ (by the definition of $N$). This is a contradiction with our assumption that $\iota(L(G_X))$ is finite, for all $X\in V$. So, there exists $X\in V$ for which $\iotaE(L(G_X))$ is infinite and, as we have seen, this means that our statement holds.
\end{proof}

So, according to Lemma~\ref{lem:one_uni_cycle}, if the $\CFG$ $G$, which is the input of our problem, contains at least one non-terminal $X \in V$ which has a $1$-universal cycle, we answer positively the instances of Problems~\ref{prob:exist_universal_largerthan_k} and~\ref{prob:universal_forall_m} defined by $G$ and, in the case of Problem~\ref{prob:exist_universal_largerthan_k}, additionally by an integer $k\geq 1$. Next, we show that one can decide in polynomial time whether such a non-terminal exists in a grammar. However, if $G$ does not contain any non-terminal with a $1$-universal cycle, while the instance of Problem~\ref{prob:universal_forall_m} is already answered negatively, it is unclear how to answer Problem~\ref{prob:exist_universal_largerthan_k}. To address this, we try to find a way to efficiently construct a string of maximal universality index, and, for that, we need another combinatorial result.

\begin{lemma}\label{lem:depthExists}
Let $G = (V, \Sigma, P, S)$ be a {\CFG} in {\CNF}, with $|V|=n$, $|\Sigma|=\sigma$, and $L = L(G)$. Furthermore, assume $\iotaE(L)$ is finite. There exists a string $w$ of $L$ with $\iota(w)=\iotaE(L)$ such that the derivation tree of $w$ has depth at most $4 n\sigma$.
\end{lemma}
\begin{proof}
Let $w_0 \in L$ be a string such that $\iota(w_0)=\iotaE(L)$, and let $T_0$ be its derivation tree. Assume that $T_0$ has depth greater than $4n\sigma$. Then there exists a simple-path $p$ in $T_0$ from the root to a leaf of length at least $4n\sigma+1$ (i.e., contains $4n\sigma +2$ nodes on it). By the pigeonhole-principle, there is one non-terminal $A\in V$ which occurs at least $4\sigma $ times on this path. Therefore, there exists the derivation $S\Rightarrow^* v_0 A v'_0 \Rightarrow^* v_0v_1 Av'_1 v'_0\Rightarrow^* \ldots \Rightarrow^* v_0v_1\cdots v_{4\sigma -1} Av'_{4\sigma-1} \cdots v'_1v'_0 \Rightarrow^*  v_0v_1\cdots v_{4\sigma-1} w'_0v'_{4\sigma-1} \cdots v'_1v'_0 =w_0$, with $v_0,v'_0,\ldots,v_{4\sigma-1},v'_{4\sigma-1},w'_0\in \Sigma^*$.

As $\iotaE(L)$ is finite, by Lemma~\ref{lem:one_uni_cycle}, $A$ has no $1$-universal cycle, so $\iota(v_1\cdots v_{4\sigma-1})=\iota(v'_{4\sigma-1}\cdots v'_{1})=0$. 

We now go with $i$ from $1$ to $4\sigma-1$ and construct a set $M_\ell$ as follows. 
For this we use of the rest of the arch factorization of a word $r(\cdot)$, which is the suffix not associated with any of the arches of the respective word.
We maintain a set $U$, which is initialized with $\alph(r(v_0))$; we also initialize $M_\ell=\emptyset$. Then, when considering $i$, if $\alph (v_i)\not\subseteq U$, we let $U\leftarrow U\cup \alph(v_i)$ and $M_\ell \leftarrow M_\ell \cup\{i\}$; before moving on and repeating this procedure for $i+1$, if $U=\Sigma$, we set $U\leftarrow \emptyset$. 
Let us note that, during this process, because $\iota(v_1\cdots v_{4\sigma-1})=0$, we set $U\leftarrow \emptyset$ at most once. Also, since $M_\ell$ is updated only when $\alph (v_i)\not\subseteq U$, it means that $M_\ell$ is updated at most $2\sigma - 2$ times. So $|M_\ell|\leq 2\sigma-2$.

Similarly, to construct a set $M_r$, for $i$ from $4\sigma-1$ downto $1$,  we maintain a set $U$, which is initialized with $ \alphabet(r(v_0 v_1\cdots v_{4\sigma-1}w'_0))$; we also initialize $M_r=\emptyset$. Then, when considering $i$, if $\alph (v'_i)\not\subseteq U$, we let $U\leftarrow U\cup \alph(v'_i)$ and $M_r \leftarrow M_r \cup\{i\}$; before moving on and repeating this procedure for $i-1$, if $U=\Sigma$, we set $U\leftarrow \emptyset$. As before, we get that $M_r$ is updated at most $2\sigma - 2$ times, and $|M_r|\leq 2\sigma-2$.

It is worth noting that the indices stored in $M_\ell$ and $M_r$ indicate the strings $v_i$ and $v'_i$, respectively, which contain letters that are relevant when computing the arch factorization of $w_0$. The indices not contained in these sets indicate strings $v_i$ or $v'_i$, respectively, which are simply contained in an arch, and all the letters of these strings already appeared in that arch before the start of $v_i$ and $v'_i$, respectively.  

As $|M_\ell|+|M_r|\leq 4\sigma-4$, we get that there exists $i\in[1,4\sigma ]$ such that $i\notin M_\ell\cup M_r$. 
Note that the derivation
$ S\Rightarrow^* v_0 A v'_0 \Rightarrow^* v_0v_1 Av'_1 v'_0\Rightarrow^* \ldots \Rightarrow^* v_0v_1\cdots v_{i-1} Av'_{i-1} \cdots v'_1v'_0 \Rightarrow^* v_0\cdots v_{i-1}v_{i+1} Av'_{i+1}v'_{i-1} \cdots v'_0
 \Rightarrow^*  v_0\cdots v_{i-1}v_{i+1} \cdots v_{4\sigma-1} w'_0v'_{4\sigma-1} \cdots  v'_{i+1}v'_{i-1} \cdots v'_0 =w_1 $
%
%
produces a string $w_1$ such that $\iota (w_1) =\iota (w_0)$; 
let $T_1$ be the tree corresponding to this derivation. Clearly, the total length of the simple-paths connecting the root to leaves in the derivation tree $T_1$ is strictly smaller than the total length of the simple-paths connecting the root to leaves in the tree $T_0$. If $T_1$ still has root-to-leaf simple-paths of length at least $4n\sigma$, we can repeat this process and obtain a tree where the total length of the simple-paths connecting the root to leaves is even smaller. This process is repeated as long as we obtain trees having at least one root-to-leaf simple-path of length at least $4n\sigma$. Clearly, this is a finite process, whose number of iterations is bounded by, e.g., the sum of the length of root-to-leaf simple-paths of $T_0$. When we obtain a tree $T$ where all root-to-leaf simple paths are of length at most $4n\sigma$, we stop and note that the border of this tree is a string $w$, with $\iota(w)=\iotaE(L)$. This concludes our proof.
\end{proof}

We now come to the algorithmic consequences of our combinatorial lemmas. For both considered problems the language given as input is in the form of a {\CFG} $G=(V, \Sigma, P, S)$ in {\CNF} with $|\Sigma| = \sigma\geq 2$ and $|V| = n$. Firstly, we show that Problem~\ref{prob:universal_forall_m} can be decided in polynomial time.

\begin{theorem}\label{thm:prob_universal_forall_m}
Problem~\ref{prob:universal_forall_m} can be solved in $\bigo(\max(n^3,n^2\sigma))$ time.
\end{theorem}
\begin{proof}
By Lemma~\ref{lem:one_uni_cycle}, it is enough to check whether $G$ contains a non-terminal $X \in V$ such that $X$ has a $1$-universal cycle. More precisely, we want to check if there exists a non-terminal $X$ such that $X\Rightarrow^* w X w'$, where $\alph(w)=\Sigma$ or $\alph(w')=\Sigma$. 

In the following, we show how to decide if there exists a non-terminal $X$ such that $X\Rightarrow^* w X w'$, where $\alph(w)=\Sigma$ (the case when $\alph(w')=\Sigma$ being similar). Recall that all non-terminals of $G$ are useful (they can be reached from the starting symbol, and a terminal string can be derivated from them). We further note that such a non-terminal $X \in V$ exists if and only if $G$ contains, for some non-terminal $X$, derivations $X\Rightarrow^* w_a X w'_a$, with $w_a\in \Sigma^*a\Sigma^*$ and $w'_a\in \Sigma^*$, for all $a\in \Sigma$. 

For this we construct a series of data structures. 

First, we construct an $n\times n$ matrix $M$, indexed by the non-terminals of $G$, where $M[A][B]=1$, for some $A,B\in V$, if and only if there exists a derivation $A\Rightarrow^* \alpha B \beta$, where $\alpha, \beta\in \Sigma^*$; otherwise, $M[A][B]=0$. This matrix can be trivially computed in $\bigo(|V|^3)$. Basically, we define a relation $R$ over $V$, where $(A,B)\in R$ if and only if $A\rightarrow BC$ or $A\rightarrow CB$ are productions of $G$ (for some non-terminal $C$). Constructing this relation takes, in the worst case $\bigo(n^3)$ time, where $n$ is the number of non-terminals. Then, we compute the transitive closure of this relation, using the Floyd-Warshall algorithm, which takes $\bigo(n^3)$ time.

Secondly, we construct an $n\times \sigma $ matrix $M'$, indexed by the non-terminals and terminals of $G$, respectively, where $M'[A][a]=1$, for some $A\in V$ and $a\in \Sigma$, if and only if there exists a derivation $A\Rightarrow^* \alpha a \beta$, where $\alpha, \beta\in \Sigma^*$; otherwise, $M'[A][a]=0$. This matrix can be trivially computed in $\bigo(n^2 \sigma )$. Basically, we first initialize all elements of $M'$ with $0$. Then, we set $M'[A][a]\gets 1$ whenever $A\rightarrow a$ is a production of $G$. Then, for all $A,B\in V$ and $a\in \Sigma$, if $M[A][B]=1$ and $B\rightarrow a$ is a production of $G$, we set $M'[A][a]\gets 1$ as well. 

The correctness of the constructions above is straightforward (given that a terminal string can be constructed from every non-terminal).

Thirdly, we construct an $n\times n$ matrix $L$, indexed by the non-terminals of $G$, where $L[A][B]=1$, for some $A,B\in V$, if and only if there exist a production $A\rightarrow BC$ in $G$ and a derivation $C\Rightarrow^* \alpha A \beta$, where $\alpha, \beta\in \Sigma^*$; otherwise, $L[A][B]=0$. This can be computed in $\bigo(n^3)$ time, in the worst case. We initialize all elements of $L$ with $0$ and, then, for every production $A\rightarrow BC$, we set $L[A][B]\gets 1$ if and only if $M[C][A]=1$. Again, the correctness of the construction is immediate.

Finally,  we construct an $n\times \sigma $ matrix $L'$, indexed by the non-terminals and terminals of $G$, respectively, where $L'[A][a]=1$, for some $A\in V$ and $a\in \Sigma$, if and only if there exists a derivation $A\Rightarrow^* \alpha a \beta A \gamma $, where $\alpha, \beta, \gamma \in \Sigma^*$; otherwise, $L'[A][a]=0$. This can be computed in $\bigo(n^2 \sigma )$ time, in the worst case. We initialize all elements of $L'$ with $0$. Then, for every non-terminal $A\in \Sigma$ and every terminal $a\in V$, we set $L'[A][a]\gets 1$ if and only if there exists a non-terminal $B$ with $M'[B][a]=1$ (i.e., there exists a derivation $B\Rightarrow^* \alpha a \beta'$)  and $L[A][B]\gets 1$ (i.e., there exists a production $A\rightarrow BC$ in $G$ and a derivation $C\Rightarrow^* \beta'' A \gamma$). The correctness of the construction follows from the explanations given alongside it.

After $L'$ is computed, we conclude that there exists a non-terminal $X \in V$ for which we have derivations $X\Rightarrow^* w_a X w'_a$, with $w_a\in \Sigma^*a\Sigma^*$ and $w'_a\in \Sigma^*$, for all $a\in \Sigma$, if and only if there exists such a non-terminal $X$ with $L'[X][a]=1$, for all $a\in \Sigma$. This can be checked in $\bigo(n\sigma)$ time.

In conclusion, our approach uses $\bigo(\max(n^3,n^2\sigma))$ time to check whether $G$ contains a non-terminal $X \in V$ such that $X$ has a $1$-universal cycle, and the statement follows. 
\end{proof}

Further, we show that Problem~\ref{prob:exist_universal_largerthan_k} is {\FPT} w.r.t.\ the parameter $\sigma$; this also means that the respective problem is solvable in polynomial time for constant-size alphabets. Recall that there is an {\ETH}-conditional lower bound of $2^{o(\sigma)}\poly(n,\sigma)$ for the time complexity of algorithms solving this problem.

\begin{theorem}\label{thm:prob_exist_universal_largerthan_k}
    Problem~\ref{prob:exist_universal_largerthan_k} can be solved in $\bigo{( 2^{4\sigma} n^5 \sigma^2 )}$ time.
\end{theorem}
\begin{proof}
Recall that now we also get as input a positive integer $k$ (given in binary representation). 

To solve Problem~\ref{prob:exist_universal_largerthan_k}, we first check, using the algorithm from Theorem~\ref{thm:prob_universal_forall_m}, whether $\iotaE(L)$ is finite. If $\iotaE(L)$ is infinite, then we answer the given instance of the problem positively. Otherwise, we proceed as follows.

We use a dynamic programming approach, to compute the maximal universality index of a string of $L$. This essentially uses the result of Lemma~\ref{lem:depthExists} which states that such a string is the border of a derivation tree of depth at most $N=4n\sigma$. 

We construct a $4$-dimensional matrix $M[\cdot,\cdot,\cdot,\cdot]$, with elements $M[i,A,S^A_p,S^A_s]$ with $A \in V$, $S^A_p, S^A_s \subsetneq \Sigma$ and $i \leq N$. By definition, $M[i,A,S^A_p,S^A_s]=\ell$ if $\ell$ is the maximum number with the property that there exists a string $w$, which labels the border of a derivation tree of height at most $i$ rooted in $A$, so that $w$ starts with a prefix $x$, with $\alph(x) = S^A_p$, followed by $\ell$ arches, and a suffix $y$ with $\alph(y) = S^A_s$.

We explain how the elements of this matrix are computed. We initialize all entries with $-\infty$. 

In the base case, for $i=1$ we only need to consider direct productions $(A,a) \in P$, for $A \in V$ and $a \in \Sigma$. We set 
$M[1, A, \{a\}, \emptyset] \gets 0 \text{ and } M[1, A, \emptyset, \{a\}] \gets 0.$ 

This means that we consider the case where the symbol $a$ is in the prefix and the case where it is in the suffix. Clearly, in each of these cases, the other set is empty, as we only have one letter. Accordingly, the entry itself is also $0$, because a symbol alone does not form an arch.

In the inductive step, we have that the depth of the considered derivation trees is at most $i>1$. We want to compute elements of the form $M[i, A, S^A_p, S^A_s]$. Our algorithm considers several cases (implemented as successive phases), and simply sets $M[i, A, S^A_p, S^A_s]$ to be the maximum in all these cases. Initially, all elements $M[i, A, S^A_p, S^A_s]$ are set to $M[i-1, A, S^A_p, S^A_s]$.

Our algorithm iterates first over productions $A\rightarrow BC$.
For a production $A\rightarrow BC$ we analyse the trees whose topmost level is defined by the respective production. Hence, we combine derivation of trees of height at most $i-1$, with roots $B$ and $C$ respectively, to obtain derivation trees of height $i$. This computation is structured as follows.

We iterate over all subsets $S_1,S_2,S_3,S_4\subsetneq \Sigma$. 

Assume first that $M[i-1,B,S_1,S_2]=0$ and $M[i-1,C,S_3,S_4]=0$ and $S_1\cup S_2\cup S_3\cup S_4 \subsetneq \Sigma$. Then, if $M[i,A,S_1\cup S_2\cup S_3\cup S_4,\emptyset]=-\infty$, we set $M[i,A,S_1\cup S_2\cup S_3\cup S_4,\emptyset]=0$, and, if $M[i,A,\emptyset, S_1\cup S_2\cup S_3\cup S_4]=-\infty$, we set $M[i,A,\emptyset, S_1\cup S_2\cup S_3\cup S_4]=0$. Basically, we can see the border of the derivation tree we have just obtained (rooted in $A$) as an empty word with $0$ arches preceded (respectively, succeeded)  by a prefix (respectively, a suffix) which contains the entire border-word. That is, similarly to the case $i=1$, we either move all letters into the prefix or suffix of the border. 
Further, we deal with different splits of the border according to the sets $S_1, S_2, S_3$, and $S_4$, retaining the property that we only create trees whose borders have $0$ arches, consisting in the empty word, preceded and succeeded by prefixes and, respectively, suffixes which are not $1$-universal. So, the following cases are considered. If  $M[i-1,B,S_1,S_2]=0$ and $M[i-1,C,S_3,S_4]=0$ and $S_1\cup S_2\cup S_3 \subsetneq \Sigma$, and $M[i,A,S_1\cup S_2\cup S_3, S_4]=-\infty$, we set $M[i,A,S_1\cup S_2\cup S_3, S_4]=0$. 
If $M[i-1,B,S_1,S_2]=0$ and $M[i-1,C,S_3,S_4]=0$ and $S_1\cup S_2\subsetneq \Sigma$, $S_3\cup S_4 \subsetneq \Sigma$, and $M[i,A,S_1\cup S_2, S_3\cup S_4]=-\infty$, we set $M[i,A,S_1\cup S_2, S_3\cup S_4]=0$. 
Finally, if $M[i-1,B,S_1,S_2]=0$ and $M[i-1,C,S_3,S_4]=0$ and $S_2\cup S_3\cup S_4 \subsetneq \Sigma$, and  $M[i,A,S_1, S_2\cup S_3\cup S_4]=-\infty$, we set $M[i,A,S_1, S_2\cup S_3\cup S_4]=0$. 

Note that the elements of $M[i,\cdot,\cdot,\cdot]$ set in this step might still be updated in the following. Basically, we have just considered the case when joining two trees with $0$ arches leads to a tree with $0$ arches. The case when we can join two trees of height at most $i-1$ whose borders have $0$ arches, and obtain a tree of height $i$ whose border has one arch is also considered in the following.

Once the above analysis is performed, we proceed as follows.

If $S_2\cup S_3\subsetneq \Sigma$, we set $T=M[i-1,B,S_1,S_2]+M[i-1,C,S_3,S_4]$, and, if $T>M[i,A,S_1,S_4]$, we update $M[i,A,S_1,S_4]=T$. In this case, we have simply joined the two subtrees (of roots $B$ and $C$), counted their respective arches, and assumed, on the one hand, that no new arch appears, and that the prefix (respectively, suffix) preceding (respectively, succeeding) the sequence of arches has alphabet $S_1$ (respectively, $S_4$). If, moreover, $S_1\cup S_2\cup S_3\subsetneq \Sigma$ and $M[i-1,B,S_1,S_2]=0$, then we set $T=M[i-1,C,S_3,S_4]$ and if $T>M[i,A,S_1\cup S_2\cup S_3,S_4]$, we update $M[i,A,S_1\cup S_2\cup S_3,S_4]=T$. In this case, we have joined the two subtrees (of roots $B$ and $C$), counted the arches originating in the subtree of root $C$, and prepended the border of the subtree of root $B$ to the prefix preceding the sequence of arches of the subtree rooted in $C$. Symmetrically, if, $ S_2\cup S_3\cup S_4\subsetneq \Sigma$ and $M[i-1,C,S_3,S_4]=0$, then we set $T=M[i-1,B,S_1,S_2]$ and if $T>M[i,A,S_1, S_2\cup S_3\cup S_4]$, we update $M[i,A,S_1, S_2\cup S_3\cup S_4]=T$. 

If $S_2\cup S_3= \Sigma$, we set $T=M[i-1,B,S_1,S_2]+M[i-1,C,S_3,S_4]+1$, and, if $T>M[i,A,S_1,S_4]$, we update $M[i,A,S_1,S_4]=C$. In this case, we have joined the two subtrees (of roots $B$ and $C$), counted their respective arches, and noticed that a new arch appears. The prefix (respectively, suffix) preceding (respectively, succeeding) the sequence of arches of the resulting tree has alphabet $S_1$ (respectively, $S_4$). 

In general, joining the two subtrees corresponding to $M[i-1,B,S_1,S_2]$ and $M[i-1,C,S_3,S_4]$, respectively, might lead to a tree with larger universality than $M[i-1,B,S_1,S_2]+M[i-1,C,S_3,S_4]+1$. But we do not need to address this situation explicitly: we always update the values $M[i,A,S^A_p,S^A_s]$ when we can obtain a tree with more arches and same prefix- and suffix-alphabet, so this situation will be covered by combining strings corresponding to $M[i-1,B,S'_1,S'_2]$ and $M[i-1,C,S'_3,S'_4]$ for some choice of the subsets $S'_1,S'_2,S'_3,S'_4$, other than the currently considered $S_1,S_2,S_3,S_4$.   

The inductive step is repeated for $i$ from $2$ to $4n\sigma$. Then, to determine $\iotaE(L)$, we take the maximum $Q$ over the entries of $M[4n\sigma, S, \emptyset , T]$, over all subsets $T\subsetneq \Sigma$, as we only consider strings $w$ that lie in $L$, so strings that can be derived from $S$. 

To solve the given instance of Problem~\ref{prob:exist_universal_largerthan_k}, we compare $Q$ with $k$. 

Accordingly, if 
$ k \leq Q=\max_{T\subseteq \Sigma } (M[4m \sigma, S, \emptyset , T])$, then we answer the respective instance positively. Otherwise, we answer it negatively. 

The correctness of the above algorithm is rather straightforward. To simplify the presentation, we say that a derivation tree with root $A$, of depth at most $i$, whose border $w$ starts with a prefix $x$, with $\alph(x) = R_1\subsetneq \Sigma$, followed by several ($0$ or more) $1$-universal words, and then a suffix $y$ with $\alph(y) = R_2\subsetneq \Sigma$, has signature $(A,i,R_1,R_2)$. 

In our algorithm, for some $i$, for some nonterminal $A$ and subsets $S^A_1,S^A_2$ of $\Sigma$, we construct in a bottom-up fashion the derivation trees rooted in $A$, of depth at most $i$, whose border $w$ starts with a prefix $x$, with $\alph(x) = S^A_1$, followed by as many $1$-universal words as possible, and a suffix $y$ with $\alph(y) = S^A_2$ (so the tree with signature $(A,i,S^A_1,S^A_2)$ with the maximum number of arches between its prefix with alphabet $S^A_1$ and its suffix with alphabet $S^A_2$). 

The construction is trivial for $i=1$: our tree corresponds to a production $A\rightarrow a$, it has $0$ arches, and either $S^A_1=\emptyset$ and $S^A_2=\{a\}$ or viceversa. 

For $i>1$, there are several cases to be analysed. The resulting tree is then obtained by taking the maximum over all these cases. 

A basic case which we need to consider is when we have that the tree with signature $(A,i,S^A_1,S^A_2)$ with the maximum number of arches between its prefix with alphabet $S^A_1$ and its suffix with alphabet $S^A_2$ is the same as the corresponding tree with signature $(A,i-1,S^A_1,S^A_2)$. Otherwise, we look for a tree with signature $(A,i,S^A_1,S^A_2)$ and depth exactly $i$. In this case, the root $A$ is rewritten by a production. Again, we must check all productions and see which leads to a tree with the desired signature, and choose the tree with the maximum number of arches, over all these possibilities. Let us now fix the production rewriting the root, namely $A\rightarrow BC$. The tree of depth $i$ with signature $(A,i,S^A_1,S^A_2)$ is obtained by joining a subtree of depth at most $i-1$ rooted in the node $B$ which has signature $(B,i-1,S_{1},S_2)$, and a subtree of depth at most $i-1$ rooted in the node $C$ which has signature $(C,i-1,S_3,S_4)$. Our algorithm considers all possible cases for these subtrees, by iterating over all choices of $S_1,S_2,S_3,S_4$, joining these trees, and computing the alphabet $S^A_1$ of the prefix and the alphabet $S^A_2$ of the suffix, as well as the number of arches of the resulting tree. The key observation is that the tree with signature $(A,i,S^A_1,S^A_2)$ with the maximum number of arches between its prefix with alphabet $S^A_1$ and its suffix with alphabet $S^A_2$ is obtained by joining some subtrees with signatures $(B,i-1,S_{1},S_2)$ and $(C,i-1,S_3,S_4)$, for some subsets $S_1,S_2,S_3$, and $S_4$, which have, respectively, a maximum number of arches between the prefix with alphabet $S_1$ (respectively, $S_3$) and the suffix with alphabet $S_2$ (respectively, $S_4$). Any other choice still leads to a tree with signature $(A,i,S^A_1,S^A_2)$, but a lower number of arches between its prefix with alphabet $S^A_1$ and its suffix with alphabet $S^A_2$. 

As far as the complexity of this approach is concerned, let us fix some $i \in [4\sigma n]$. For $i=1$, the algorithm runs in $\bigo{(n\sigma)}$, in the worst case. For $i > 1$, the inductive step of the dynamic programming (covering all cases and their subcases) runs in $\bigo{( n^4 2^{4\sigma} \sigma})$ time, in the worst case. 

Indeed, we consider all productions $A\rightarrow BC$. Then, we go over all choices of the four sets $S_1,S_2,S_3,S_4$, and use these together with the values of the elements $M[i-1, B, \cdot, \cdot]$ and $M[i-1, C, \cdot, \cdot]$, to determine/update the values of the entries $M[i, A, \cdot, \cdot]$. Depending on the case in which we are, we need to check if the union of some of these sets equals $\Sigma$ (which takes $\bigo (\sigma)$ time), and, in all cases, we need to sum up some elements of $M[i-1,\cdot,\cdot,\cdot]$ and compare them with the current elements stored in $M[i,\cdot,\cdot,\cdot]$. Now, we note that the numbers stored in the matrix $M$ can become too big to assume that such arithmetic operations and comparisons can be done in constant time. For instance, a string generated by a $\CFG$ in $\CNF$ with a derivation whose tree has depth $4n\sigma$ may have $\Theta(2^{4n\sigma})$ letters and therefore also $\Theta(2^{4n\sigma}/\sigma)$ arches, in the worst case. Working with such big numbers may add a factor $\bigo(n\sigma)$ to our complexity. So, to wrap this up, for some $i$, we need to iterate over all productions, and choice of sets $S_1,S_2,S_3,S_4\subseteq \Sigma$, and for each choice of these parameters, we do at most $\bigo(n\sigma)$ steps. This takes $\bigo{( n^4 2^{4\sigma} \sigma})$ time overall.

As already mentioned, we iterate this process $4n\sigma$ times, which leads to an overall time complexity of $\bigo{( 2^{4\sigma} n^5 \sigma^2 )}$ for our algorithm.
\end{proof}

The algorithm from Theorem~\ref{thm:prob_exist_universal_largerthan_k} uses exponential space (due to the usage of the matrix $M$). However, there is also a simple (non-deterministic) PSPACE-algorithm solving this problem. Such an algorithm constructs non-deterministically the left derivation (where, at each step, the leftmost non-terminal is rewritten) of a string $w\in L$ with $\iota(w)\geq k$; $w$ is non-deterministically guessed, and it is never constructed or stored explicitly by our algorithm. During this derivation of $w$ the number of non-terminals in each sentential form is upper bounded by the depth of its derivation tree~\cite{HopcroftU79}; due to Lemma~\ref{lem:depthExists}, we thus can have only $4n\sigma $ such non-terminals (if this number becomes larger, we stop and reject: the derivation tree of the guessed derivation is too deep for our purposes). During the simulation of the leftmost derivation, at step $i$, we also do not keep track of the maximal prefix $w'_i$ consisting only of terminals of the sentential form, but only of $\iota(w'_i)$, $\alph(r(w'_i))$, and of the maximal suffix $w''_i$ consisting of non-terminals only (i.e., the part we still need to process); this is enough for computing the universality of the derived string. The information stored by our algorithm clearly fits in polynomial space. If, and only if, at the end of the derivation, the maintained universality index is at least $k$, we accept the input grammar and number $k$. 

 \section{Problem~\ref{prob:exist_nonuniversal}}\label{sec:Problem4}

Let us note that deciding Problem~\ref{prob:exist_nonuniversal} for some input language $L$ and integer $k$ is equivalent to deciding whether $\iotaA (L)\geq k$. In \cite{bib:UniReg}, it was shown that for a regular language $L$ over an alphabet with $\sigma$ letters, accepted by an NFA with $s$ states, Problem~\ref{prob:exist_nonuniversal} can be decided in $\bigo(s^3\sigma)$. 

For the rest of this section, we consider Problem~\ref{prob:exist_nonuniversal} for the class $\CFL$, and we assume that we are given a {\CFL} $L$ by a {\CFG} $G$ in $\CNF$, with $n$ non-terminals, over an alphabet $\Sigma$, with $\sigma\geq 2$ letters. Recall that our approach is to compute $\iotaA(L)$ and compare it with the input integer $k$. 

As before, we start with a combinatorial observation. Intuitively, when we try to find a word with the lowest universality index, it is enough to consider words $w$, whose derivation trees do not contain root-to-leaf paths which contain twice the same non-terminal (otherwise, such a tree could be reduced, to a derivation tree of a word with potentially lower universality index). 

\begin{lemma}\label{lem:depthForAll}
If $w\in L$ is a string with $\iota(w)\leq \iota(w'),$ for all $w'\in L$, then there exists a string $w''\in L$ with $\iota(w'')=\iota(w)$ and the derivation tree of $w''$ has depth at most $n$. 
\end{lemma}
\begin{proof}
Let us consider a derivation tree for $w$. If this derivation tree of $w$ has depth greater than $n$, then there would be a simple-path in this tree, from the root to a leaf, which contains the same non-terminal $A$ twice. Hence, we have the derivation $S\Rightarrow^* v A v' \Rightarrow^* v u A u' v'\Rightarrow v u z u' v' = w$. But, then we have the derivation $S\Rightarrow^* v A v' \Rightarrow^*  v  z v' = w_0$, and, as $w_0$ is a subsequence of $w$ we have that $\iota(w_0)\leq \iota(w)$. If the derivation tree associated to this derivation has depth at most $n$, then we can take $w''=w_0$ and are done. Otherwise, we repeat the process with $w_0$ in the role of $w$ until we get a derivation whose associated tree has depth at most $n$ (this process is finite, as in each iteration of this process, we decrease the length of at least one path). 
\end{proof}

We now show that we can compute $\iotaA(L)$ in polynomial time, when the input language is a $\CFL$.
\begin{theorem}\label{thm:Problem4}
Problem~\ref{prob:exist_nonuniversal} can be solved in $\bigo(n^4\sigma^2)$ time.
\end{theorem}
\begin{proof}
To compute $\iotaA(L)$, we compute the smallest integer $\ell$ for which there exists a string $w\in L$ which has an absent subsequence of length $\ell$ (and then conclude that $\iotaA(L)=\ell-1$). Indeed, this holds because we want to compute the universality index of a word $w\in L$ with the smallest universality, which, as the shortest absent subsequence of some word $x$ has length $\iota(x)+1$, is equivalent to computing the length of the shortest absent subsequence of some word $w\in L$. 

Our approach to computing $\iotaA(L)$ is, thus, to define a $4$-dimensional matrix $M$ whose elements are $M[i,A,a,b]$, with $i\in [n]$, $A\in V$, $a\in \Sigma\cup\{\varepsilon\}, b\in \Sigma$. By definition, $M[i,A,a,b]=\ell$ if and only if $\ell $ is the smallest integer for which:
\begin{itemize}
    \item there exists $w\in \Sigma^*$ such that $a$ occurs in $w$, $A\Rightarrow^* w$ and this derivation has an associated tree of depth at most $i$, and 
    \item there exists a string $v$ of length $\ell$ starting with $a$ and ending with $b$ such that $v$ is an absent subsequence of $w$.
\end{itemize}
That is, $v$ is a shortest string which starts with $a$ and ends with $b$ and is an absent subsequence for some string of the language. This means that $\ell$ is the length of the shortest word which is $\SAS_{a,b}(w)$, over all $w\in \Sigma^*$ which can be derived from $A$ by a derivation tree of depth at most $i$.

The elements of the matrix $M$ are computed as follows. 

Initially, all the elements of this matrix are set to $+\infty$. 

In the base case, for $i=1$, for a production $A\rightarrow a$, we set, on the one hand, $M[1,A,a,b] \gets 2$, for all $b\in \Sigma$ (as $ab$ is not a subsequence of $a$, clearly, and there are no shorter strings which start with $a$ and are not a subsequence of $a$) and, on the other hand, $M[1,A,\varepsilon,b] \gets 1$, for all $b\in \Sigma\setminus \{a\}$ (as $b$ is not a subsequence of $a$, clearly, and there are no shorter strings which start with $\varepsilon$ and are not a subsequence of $a$). The rest of the elements of the form $M[1,\cdot,\cdot,\cdot]$ remain equal to $+\infty$, as we have considered all derivation trees of depth $1$. 

In the inductive step, we consider some $i>1$ and we compute the elements of the form $M[i,\cdot,\cdot,\cdot]$, based on the elements $M[i-1,\cdot,\cdot,\cdot]$.

Let us fix the non-terminal $A$, and we explain how the elements $M[i,A,\cdot,\cdot]$ are computed. 

Clearly, $M[i,A,a,b]\leq M[i-1,A,a,b]$, for all $A\in V, i\in [n], a\in \Sigma\cup\{\varepsilon\}, b\in \Sigma$; that is, alongside the trees considered in the previous step, we need to consider in this step the derivation trees of depth exactly $i$, where $i>1$. To this end, note that a tree of depth $i>1$ with root $A$ corresponds to a derivation that starts with a production $A\rightarrow BC \in P$, for some $B,C \in V$.

Moreover, we note that, for $C\in V, c\in \Sigma$, $M[i-1,C,\varepsilon,c]=1$ if and only if there is some word which can be derived from $C$, with a derivation whose tree has depth $i-1$, and contains no $c$. 

Let us now focus on the computation of $M[i,A,a,b]$, for some $a\in \Sigma\cup\{\varepsilon\}$ and $b\in \Sigma$.

We first assume that $a=\varepsilon$. We need to consider several cases and we refer to the definition of $M[i,A,a,b]$ given above. Basically, a string $w$ that fits the respective definition, and is the border of a tree of depth $i>1$, must be of the form $w=w_1w_2$, such that there exists a production $A\rightarrow BC\in P$ and $w_1$ is derived from $B$ and $w_2$ is derived from $C$. We do not have any restriction on the starting letter of $w$, so, in the first case, an $\SAS_{\varepsilon,b}(w)$ could be obtained as the concatenation of an $\SAS_{\varepsilon,x}(w_1)$ from which we remove the final letter $x$, for some letter $x\in \alph(w_2)$, and an $\SAS_{x,b}(w_2)$ (note that $x$ can also be equal to $b$). In the second subcase, which only occurs when $b\notin\alph(w_2)$, an $\SAS_{\varepsilon,b}(w)$ can be obtained as an $\SAS_{\varepsilon,b}(w_1)$ (as such a string is also absent from $w_1w_2$). Clearly, one of these two cases must occur - there is no other possibility to define an $\SAS_{\varepsilon,b}(w)$. We can now compute $M[i,A,\varepsilon,b]$.

So, for each production $A\rightarrow BC\in P$, we proceed as follows. 

Firstly, we define $f^1_{B,C} \gets  \min\limits_{x \in \Sigma, M[i-1,C,\varepsilon,x]>1} ( M[i-1,B,\varepsilon,x] + M[i-1,C,x,b]-1)$. This corresponds to the first case above.

Then, if $M[i-1,C,\varepsilon,b]=1$, we define $f^2_{B,C} \gets M[i-1,B,\varepsilon,b] $; otherwise, we set$f^2_{B,C} \gets +\infty$. This corresponds to the case when $w_2$ contains no letter $b$, the second case from the discussion above. 

Finally, we set $M[i,A,a,b]\gets \min (\{M[i-1,A,a,b]\} \cup \{f^1_{B,C},f^2_{B,C}\mid A\rightarrow BC\in P\})$. The correctness of this formula follows from the explanations given above.

Further, assume $a\neq \varepsilon$. We again need to consider several cases and we refer to the definition of $M[i,A,a,b]$ given above. Basically, a string $w$ that fits the respective definition, and is the border of a tree of depth $i>1$, must be of the form $w=w_1w_2$, such that there exists a production $A\rightarrow BC\in P$ and $w_1$ is derived from $B$ and $w_2$ is derived from $C$. There are two main cases that may occur. Firstly, $w_1$ does not contain any $a$; then, an $\SAS_{a,b}(w)$ can only be obtained as an $\SAS_{a,b}(w_2)$. Secondly, if $w_1$ contains $a$, then we have again two subcases. In the first such subcase, an $\SAS_{a,b}(w)$ could be obtained as the concatenation of an $\SAS_{a,x}(w_1)$ from which we remove the final letter $x$, for some letter $x\in \alph(w_2)$, and an $\SAS_{x,b}(w_2)$ (note that $x$ can also be equal to $b$). In the second subcase, which only occurs when $b\notin\alph(w_2)$, an $\SAS_{a,b}(w)$ can be obtained as an $\SAS_{a,b}(w_1)$ (as such a string is also absent from $w_1w_2$). Clearly, one of these two cases must occur - there is no other possibility to define an $\SAS_{a,b}(w)$. Based on these cases, we can proceed to compute $M[i,A,a,b]$.

So, for each production $A\rightarrow BC\in P$, we proceed as follows.

Firstly, if $M[i-1,B,\varepsilon,a]=1$, we define $m^1_{B,C} \gets  M[i-1,C,a,b]$; otherwise, we set $m^1_{B,C}=+\infty$. This corresponds to the first case from the discussion above, when $w_1$ contains no letter $a$. 

In the next two cases we have $M[i-1,B,\varepsilon,a]>1$. 

We define $m^2_{B,C} \gets \min\limits_{x \in \Sigma, M[i-1,C,\varepsilon,x]>1} M[i-1,B,a,x] + M[i-1,C,x,b]-1$; this corresponds to the case when $w_1$ contains letter $a$, and $w_2$ contains letter $x$, the first subcase of the second case from the discussion above. 
Finally, if $M[i-1,C,\varepsilon,b]=1$, we define $m^3_{B,C} \gets M[i-1,B,a,b]$; otherwise, we define $m^3_{B,C} \gets +\infty$. This last case corresponds to the situation when $w_1$ contains letter $a$, and $w_2$ contains no letter $b$, the remaining subcase of the second case from the discussion above. 

Finally, we set $M[i,A,a,b]\gets \min (\{M[i-1,A,a,b]\} \cup \{m^1_{B,C},m^2_{B,C},m^3_{B,C}\mid A\rightarrow BC\in P\})$. The correctness of this formula follows from the explanations given above.

We repeat this inductive step, in which the elements $M[i,\cdot,\cdot,\cdot]$ are computed, for $i\in [n]$, following the result of \cref{lem:depthForAll}. 

Once all elements of $M$ are computed, we note that $\iotaA(L)$ is obtained by subtracting $1$ from the minimum element $min$ of the form $M[n,S,a,b]$, with $a\in \Sigma \cup \{\varepsilon\}$ and $b\in \Sigma$. Indeed, the respective element $min$ is the length of the shortest string $v$ for which there exists $w\in L$ such that $v$ is absent from $w$, and, if $|v|\geq 1$ then $a=v[1]$ and $b=v[min]$, while if $|v|= 1 $ then $a=\varepsilon$ and  $b=v[min]$. Clearly, $\iota(w)=min-1$, and there is no other string in $L$ with a lower universality index. This concludes the proof of the correctness for our dynamic programming approach.

The complexity of the dynamic programming algorithm described is $\bigo(n^4 \sigma^2)$ (as we have $n$ iterations, in each iteration we compute $n\sigma^2$ elements, and for each of these elements we have to iterate over $n^2$ pairs of non-terminals).
\end{proof}

 \section{What Next? Conclusions and First Steps Towards Future Work}
\label{sec:dfawtl}
    
A conclusion of this work is that the complexity of the approached problems is, to a certain extent, similar when the input language is from the classes {\REG} and {\CFL} and they all become undecidable for {\CSL}. So, a natural question is whether there are classes of languages (defined by corresponding classes of grammars or automata) between {\REG} and {\CSL} which exhibit a different, interesting behaviour. 

We commence here this investigation by considering the class of languages accepted by a model of automata, namely, the deterministic finite automata with translucent letters (or, for short, translucent (finite) automaton -- {\dfawtl}), which generalizes the classical {\dfa} by allowing the processing of the input string in an order which is not necessarily the usual sequential left-to-right order (without the help of an explicit additional storage unit). These automata, first considered in~\cite{NagOtt11} (see also the survey~\cite{Ott23} for a discussion on their properties and motivations), are strictly more powerful than classical finite automata
and are part of a class of automata-models that are allowed to jump symbols in their processing, e.g., see~\cite{JFA} or~\cite{ChiFazYam16}. From our perspective, these automata and the class of languages they accept are interesting because, on the one hand, they seem to be a generalization of regular languages which is orthogonal to the classical generalization provided by context-free languages, and, on the other hand, initial results suggest that the problems considered in this paper, not only become harder for them, but also their decidability fills the gap between the polynomial time solubility in the case of {\CFL}s and that of undecidability for the class of {\CSL}.

So, in what follows, we discuss some problems from Section~\ref{sec:Problems} in relation to the {\dfawtl} model, following the formalization from~\cite{MiPaPaSC24}.

\begin{definition} 
    A \dfawtl\ $M$ is a tuple $M=(Q,\Sigma,q_0,F,\delta)$, just as in the case of DFA. However, the processing of inputs is not necessarily sequential. We define the partial relation $\circlearrowright$ on the set $Q\times \Sigma^*$ of configurations of $M$: $(p,xay) \circlearrowright_M (q,xy)$ if $\delta(p,a)=q$,
    and $\delta(p,b)$ is not defined for any $b\in alph(x)$, where $p,q\in Q$, $a,b\in \Sigma$, $x, y\in \Sigma^*$. The subscript $M$ is omitted when it is understood from the context. The reflexive and transitive closure of $\circlearrowright$ is $\circlearrowright^*$ and the language accepted by $M$ is defined as
    $L(M)=\{w\in \Sigma^*\mid (q_0, w)\circlearrowright^* (f,\varepsilon) \mathrm{\ for\ some\ } f\in F\}.$
    \end{definition}
    In this model, letters $a$ such that $\delta(p,a)$ is not defined are called translucent for $p$, hence the name of the model. The machine reads and erases from the tape the letters of the input one-by-one. Note that the definition requires that every letter of the input is read before it can be accepted. This is slightly different from the original definition~\cite{NagOtt11}, which did not require all of the letters read, and used an unerasable endmarker on the tape. \dfawtl\ by our definition can be trivially simulated by a machine with the original definition, and our results stand for the original model, too. We chose to follow the definitions in~\cite{MiPaPaSC24}, because in our opinion it is simpler (and simpler to argue), and illustrates the difficulty of the subsequence matching problems for nonsequential machine models just as well.
    
    A first observation is that, in terms of execution, in each step a {\dfawtl} reads (and consumes) the leftmost unconsumed symbol which allows a transition (i.e., that has not been previously read, and there is a transition labeled with it from the current state). Therefore, for every individual letter, the order of the processing of its occurrences in the {\dfawtl} is that in which they appear in a string. The non-deterministic version of this automata model accepts all rational trace languages, and all accepted languages have semi-linear Parikh images. Moreover, the class of languages accepted by this model is incomparable to the class of {\CFL}, while still being {\CS}. The class of languages accepted by the more restrictive deterministic finite automata with translucent letters, for short {\dfawtl}, strictly includes the class {\REG} and is still incomparable with {\CFL} and the above mentioned class of rational trace languages. The recent survey~\cite{Ott23} overviews the extensive literature regarding variations of these types of machines.

\begin{example}
    The \dfawtl\ in Figure~\ref{fig:dfawtl_ex} accepts the language $L=w\shuffle h(w)$, where $w\in \{a,b\}^*$ and $h:\{a,b\}^* \rightarrow \{c,d\}^*$ is a morphism given by $h(a)=c$, $h(b)=d$. Here $\shuffle$ denotes the usual shuffle operation for words over some alphabet $\Sigma$, i.e., $u\shuffle v=\{u_1v_1\cdots u_\ell v_\ell\mid u=u_1\cdots u_{\ell}, v=v_1\cdots v_\ell, u_i\in \Sigma^*$ for $i\in [\ell], v_i\in \Sigma^*$ for $i\in [\ell]\}$; in our case, $\Sigma = \{a,b,c,d\}$. 
    
    Coming back to the \dfawtl\ in Figure~\ref{fig:dfawtl_ex}: in state $q_0$, the machine can read only the first $a$ or $b$ remaining on the tape and immediately matches it with the first $c$ or $d$, respectively. If it reads $a$ and in the remaining input the first $d$ comes before the first $c$, it goes into the sink state. Similarly, if it reads $b$ but the first remaining $c$ is before the first remaining $d$, it goes to sink, because the projection of the input to the $\{a,b\}$ alphabet does not match the projection to the $\{c,d\}$ alphabet. The language $L$ is not context-free. This can be, indeed, seen by intersecting it with the regular language $(a+b)^* (c+d)^*$, which yields the language $\{w\cdot h(w) \mid w\in \{a,b\}^* \}$, a variant of the so called `copy language'. This language is non-context-free, by an easy application of the Bar-Hillel pumping lemma, so $L$ is not context-free.

    \begin{figure}
        \centering
        \begin{tikzpicture}[scale=0.2]
\tikzstyle{every node}+=[inner sep=0pt]
\draw [black] (13.3,-27.9) circle (3);
\draw (13.3,-27.9) node {$q_0$};
\draw [black] (13.3,-27.9) circle (2.4);
\draw [black] (26.6,-21.2) circle (3);
\draw (26.6,-21.2) node {$q_1$};
\draw [black] (26.6,-35.2) circle (3);
\draw (26.6,-35.2) node {$q_2$};
\draw [black] (39.9,-27.9) circle (3);
\draw (39.9,-27.9) node {$sink$};
\draw [black] (6.5,-27.9) -- (10.3,-27.9);
\fill [black] (10.3,-27.9) -- (9.5,-27.4) -- (9.5,-28.4);
\draw [black] (15.98,-26.55) -- (23.92,-22.55);
\fill [black] (23.92,-22.55) -- (22.98,-22.46) -- (23.43,-23.36);
\draw (20.89,-25.05) node [below] {$a$};
\draw [black] (14.494,-25.16) arc (147.82332:85.6509:9.968);
\fill [black] (14.49,-25.16) -- (15.34,-24.75) -- (14.5,-24.22);
\draw (17.51,-21.06) node [above] {$c$};
\draw [black] (15.93,-29.34) -- (23.97,-33.76);
\fill [black] (23.97,-33.76) -- (23.51,-32.93) -- (23.03,-33.81);
\draw (20.95,-31.05) node [above] {$b$};
\draw [black] (23.698,-35.91) arc (-85.09945:-152.42287:9.712);
\fill [black] (14.26,-30.73) -- (14.19,-31.67) -- (15.07,-31.21);
\draw (17.2,-35.25) node [below] {$d$};
\draw [black] (29.28,-22.55) -- (37.22,-26.55);
\fill [black] (37.22,-26.55) -- (36.73,-25.74) -- (36.28,-26.64);
\draw (34.24,-24.05) node [above] {$d$};
\draw [black] (29.23,-33.76) -- (37.27,-29.34);
\fill [black] (37.27,-29.34) -- (36.33,-29.29) -- (36.81,-30.17);
\draw (34.19,-32.05) node [below] {$c$};
\end{tikzpicture}
        \caption{\dfawtl\ that accepts the language $w\shuffle h(w)$, where $w\in \{a,b\}^*$ and $h$ is a morphism of the form $h(a)=c$, $h(b)=d$.}
        \label{fig:dfawtl_ex}
    \end{figure}
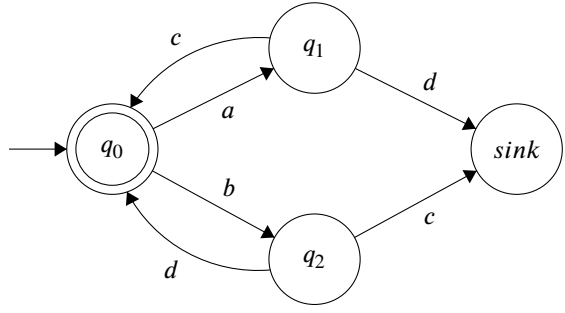
\end{example}


We first note that the class of languages accepted by {\dfawtl} becomes incomparable to that of {\CFL}s only starting from the ternary alphabet case (see~\cite{NagOtt13}), since, for a {\dfawtl} over a binary alphabet, one can construct a push-down automaton accepting the same language.

\begin{theorem}\label{thm:dfawtl-CF}
    The languages accepted by {\dfawtl} over binary alphabets are {\CF}.
\end{theorem}
\begin{proof}
    Since we deal with only two letters and a deterministic machine, at any point in the computation we must read at least one of the two symbols, and can only do a translucent transition over unary factors. The implication is that it is impossible for such a \dfawtl\ to read a letter that is preceded by a binary factor in the remaining input. The idea of the simulation is to store in the stack the letters that have been read `out of order', i.e., letters preceded by a prefix in the remaining input. The stack contents then can be matched with the first upcoming occurrences of the corresponding letter. The stack never holds two different letters, its content is always unary. Because of it being unary, it does not matter in what order the letters currently on the stack have been pushed. Whenever we have a prefix $a^kb$ of the remaining input and we are in a state $p$ for which $a$ is translucent and $\delta(p,b)=q$, we transition to $q$ without reading anything from the tape, and put $b$ on the stack to be matched later with the $b$ coming after the prefix $a^k$. Then we continue the computation from $q$, reading $a^k$. We never need to have two different letters on the stack, because to put $a$ there when there are $b$'s already, would mean that we made translucent transitions reading some $b$'s that have not been matched yet and currently we are in a state from which $b$ is translucent. If the upcoming letter on the tape is $b$, we can match it from the stack without changing states (reflecting that, in fact, this $b$ would have been read earlier by the \dfawtl). If the upcoming letter is an $a$, we can simply read it from the current state without changing the stack contents, i.e., knowing that we still need to match some $b$'s later.  
    
    Let $A=(Q,\{a,b\}, q_0,F, \delta)$ be our input {\dfawtl}. 
    
    We construct the nondeterministic PDA $B=(Q,\{a,b\},\{a,b\}, q_0, \{f_B\}, \delta')$ as follows. For each $x\in\{a,b\}$ and each state $p\in Q$ such that $\delta(p,y)=q$ but $\delta(p,x)$ is not defined, where $y\neq x$, in the PDA we define the transitions $\delta'(p,\bot,\varepsilon)=(q,y)$, $\delta'(p,y,\varepsilon)=(q,yy)$ and $\delta'(p,x,y)=(q,x)$. For each state $p$ and letter $x$, we add the transitions $\delta'(p,x,x)=(p,\varepsilon)$. For each state $p$ such that $\delta(p,a)$ and $\delta(p,b)$ are both defined in $A$ we add $\delta'(p,\bot,x)=(q,\varepsilon)$, $\delta'(p,y,x)=(q,y)$ if $\delta(p,x)=q$. 
    
    For each state $p$ of the {\dfawtl} $A$ such that $p$ cannot read $x$, when the PDA is in a state $p$, we simulate $A$'s possible translucent transitions in the PDA by pushing $y$ on the stack to be matched later with the first remaining $y$ in the input. Whenever the PDA has $x$ on the top of the stack, it means that earlier $A$ would have performed some transitions on reading $x$ from states to which $y$ was translucent, and during the simulation we have not matched the $x$ yet. In that case if it  reads $x$ from the tape, it does not change its state, but removes an $x$ from the stack to simulate that $A$ ``already read that symbol earlier''. If the PDA has an empty stack, we matched all the symbols read by $A$ by applying translucent transitions over prefixes, and the PDA can perform identical moves as $A$. Finally, for each $f\in F$, we add a transition $\delta'(f,\bot,\varepsilon)=(f_B,\varepsilon)$. This makes sure that the input is only accepted with an empty stack, i.e., when all translucent transitions have been accounted for.

    The result easily follows with the pushdown counting the translucent transitions.
\end{proof}

As a consequence of this and of the results shown in the previous sections we get the following.

\begin{theorem}\label{cor:dfawtl-binary}
Over binary alphabets, all problems of Section~\ref{sec:Problems} are decidable, and except for Problem~\ref{prob:exist_universal_largerthan_k}, all are decidable in polynomial time for a {\dfawtl} $A$ given as input.
\end{theorem}

Thus, our interest now shifts to languages accepted by {\dfawtl}s, over alphabets $\Sigma$ of size $\sigma\geq 3$. We report here a series of initial results, which suggest this to be a worthwhile direction of investigation. We first note that we cannot apply the approach from the general Theorem~\ref{thm:intersection_dec} to solve the problems considered, since one can encode the solution set of any Post Correspondence Problem (for short, \PCP) instance as the intersection of a regular language and a language accepted by a {\dfawtl}. This is a first significant difference w.r.t.\ the status of the approached problems for the case of {\REG} and {\CFL}. 

\begin{theorem}\label{thm:PCP_Prob1-dfawtl}
The emptiness problem for languages defined as the intersection of the language accepted by a {\dfawtl} with a regular language (given as finite automaton) is undecidable.
\end{theorem}
\begin{proof}
    Let the \PCP\ instance be given as the string vectors $(u_1,\dots, u_k)$ and $(v_1,\dots,v_k)$ over an alphabet $\Sigma$. We construct the {\dfawtl} $M$ as follows. The input alphabet of $M$ is $\{1,\dots,k\}\cup \bigcup_{i=1}^4\Sigma_i$, where $\Sigma_i=\{a_i \mid a\in \Sigma\}$. For each $i\in [1,4]$, let $pr_i$ be the morphism that turns strings over $\Sigma$ into strings over $\Sigma_i$ by simply attaching a subscript $i$ to each letter ($pr_i(a)=a_i$). For each $j\in [k]$, the automaton has states $p_{j}$ such that $\delta(q_0,j)=p_j$, where $q_0$ is the initial state. For each $j\in [k]$, we add paths in $M$ (disjoint from paths added from $p_\ell$, for $\ell\neq j$) such that $\delta(p_j,pr_1(u_j))=r_j$, and all strings $x\in \Sigma_1^*$  such that $x\neq pr_1(u_j)$ lead to a sink state from $p_j$. From $r_j$ we add paths (disjoint from paths added from $r_\ell$, for $\ell\neq j$) such that $\delta(r_j, pr_2(v_j))=s_j$ and all strings $x\in \Sigma_2^*$ such that $x\neq pr_2(v_j)$ lead to a sink state from $r_j$. From $s_j$ we have a single transition $\delta(s_j,\#)=q_0$. 
    The part of $M$ constructed so far reads completely inputs of the form
$$i_1\cdots i_\ell \shuffle pr_1(u_{i_1}\cdots u_{i_\ell}) \shuffle pr_2(v_{i_1}\cdots v_{i_\ell})\shuffle \#^\ell,$$ 
    returning to $q_0$, where $i_1, \ldots, i_k\in[k]$, i.e., it makes sure that the strings chosen from both string vectors correspond to the same indices, in the same order. 

    \begin{figure}
        \centering
        \includegraphics[width=0.7\linewidth]{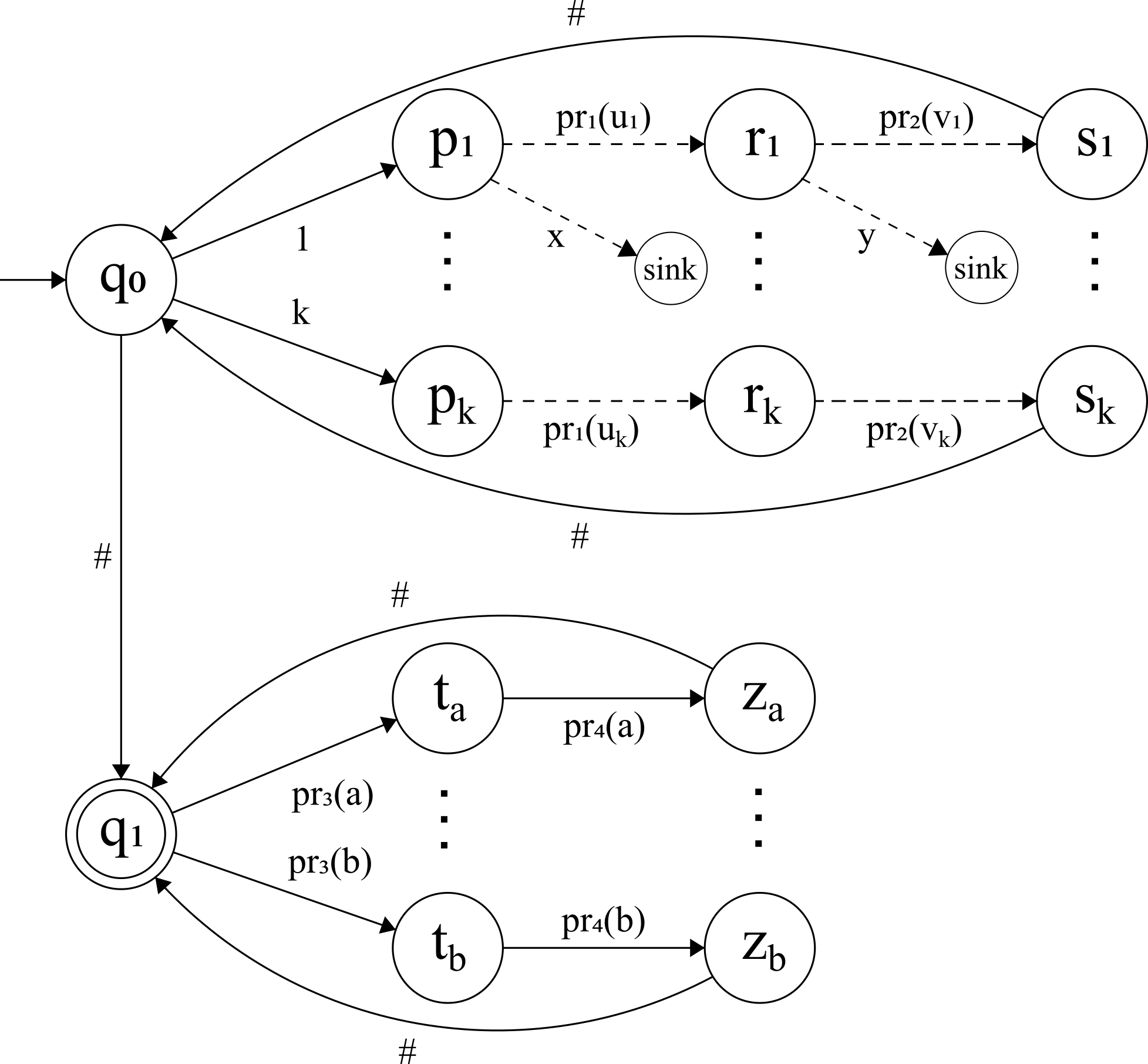}
        \caption{The \dfawtl\ accepting $i_1\cdots i_\ell \shuffle pr_1(u_{i_1}\cdots u_{i_\ell}) \shuffle pr_2(v_{i_1}\cdots v_{i_\ell}) \shuffle pr_3(w) \shuffle pr_4(w)\shuffle \#^{m}$, where $i_1,\dots,i_\ell \in [k]$, $w,u_j,v_j\in \Sigma^*$, $a,b,\dots \in \Sigma$, and $x\in \Sigma_1^*$ represents all strings over the alphabet with subscript $1$ such that $x\neq pr_1(u_1)$, and $y\in \Sigma_2^*$ represents all strings over alphabet with subscript $2$ such that $y\neq pr_2(v_1)$.}
    \label{fig:PCPdfawtl}
    \end{figure}

    The next part is for checking whether the concatenation of the strings is actually equal, i.e., $u_{i_1}\cdots u_{i_\ell}=v_{i_1}\cdots v_{i_\ell}$. We first add a transition $\delta(q_0,\#)=q_1$. From $q_1$ we also add transitions $\delta(q_1,pr_3(a))=t_a$, for each $a\in \Sigma$, and from $t_a$ we add $\delta(t_a,pr_4(a))=z_a$. From $z_a$ we return to $q_1$ by reading $\#$. We set $q_1$ to be the only final state. Now the machine (see Fig.~\ref{fig:PCPdfawtl}) accepts shuffles 
    \begin{equation}\label{eq:shuffle}
    i_1\cdots i_\ell \shuffle pr_1(u_{i_1}\cdots u_{i_\ell}) \shuffle pr_2(v_{i_1}\cdots v_{i_\ell}) \shuffle pr_3(w) \shuffle pr_4(w)\shuffle \#^{m} 
    \end{equation}
    where $i_1,\dots,i_\ell\in [k]$,  $w\in \Sigma^*$ and $m=\ell+|w|$, but it has no way of checking whether $w=u_{i_1}\cdots u_{i_\ell}=v_{i_1}\cdots v_{i_\ell}$. To check that equality we intersect $L(M)$ with the regular language 
    \[
    \left(\bigcup_{i=1}^k i\right)^* \left(\bigcup_{a\in \Sigma} pr_1(a) pr_3(a)\right)^* \left(\bigcup_{a\in \Sigma} pr_2(a) pr_4(a)\right)^* \#^*.
    \]
    that filters out most strings of the language defined by the shuffle in \cref{eq:shuffle}, leaving 
    \[
    i_1\cdots i_\ell \cdot (pr_1(u_{i_1}\cdots u_{i_\ell}) \shuffle_p pr_3(u_{i_1}\cdots u_{i_\ell})) \cdot  (pr_2(v_{i_1}\cdots v_{i_\ell}) \shuffle_p pr_4(v_{i_1} \cdots v_{i_\ell})) \cdot \#^{m}
    \]
    where $\shuffle_p$ is the perfect shuffle of two strings, i.e., $a_1\cdots a_n \shuffle_p b_1\cdots b_n=a_1b_1a_2b_2\cdots a_nb_n$.
    The resulting language consists exactly of the strings encoding solutions of the \PCP\ instance, which means that the language is empty if and only if the \PCP\ instance has a solution, a well-known undecidable problem.
\end{proof}

The decidability of Problem~\ref{prob:wtlsubseq} for larger alphabets in the case of {\dfawtl} is settled by an exponential time brute force algorithm, after establishing that if the input language contains a supersequence of $w$, then it also contains one whose length is bounded by a polynomial in the size of the input. By a reduction from the well-known NP-complete Hamiltonian Cycle Problem~\cite{GareyJ}, we can also show that Problem~\ref{prob:wtlsubseq} for {\dfawtl} is NP-hard over unbounded alphabets (containment in NP follows from the same length upper bound mentioned earlier).
This is again a significant deviation w.r.t.\ the status of this problem for the case when the input language is given by a finite automaton or by a {\CFG}.

\begin{theorem}\label{Thm:Prob1-dfawtl-NPhardness}
    Problem~\ref{prob:wtlsubseq} is NP-complete over unbounded alphabets.
\end{theorem}
\begin{proof}
    The Hamiltonian Cycle Problem is to decide, given a simple graph $G$, whether there is a cycle in $G$ that visits each vertex of $G$ exactly once. This is a well-known NP-complete problem~\cite{GareyJ}.

    To show NP-hardness, We reduce the HCP to Problem~\ref{prob:wtlsubseq}. Let $G=(\{v_1,\dots,v_n\},E)$ be an undirected simple graph. We define a {\dfawtl} $A=(Q,\{v_1,\dots,v_n,a,b\},q_0, F, \delta)$ that accepts some string that is a supersequence of $v_1\cdots v_n$ if and only if $G$ is a Hamiltonian graph. First we construct a gadget $A_1$ that tracks the first step of the Hamiltonian path and then create $n$ copies of it such that the automaton moves from gadget $A_i$ to $A_{i+1}$ if and only if the corresponding current path in $G$ can be extended by another edge. Thus the automaton accepts strings that correspond to paths of length $n$ in $G$, so $v_1\cdots v_n$ is a subsequence of an accepted string if and only if there is a closed path of length $n$ starting from $v_1$ reaching all vertices.

    Within each gadget $A_i$ with $i\in [1,n]$ we have states $v_{i,j}$ and $w_{i,j}$, for each $j\in [1,n]$. We also have additional states in the gadget that allow, for each $j,k\in [1,n]$, to transition from $v_{i,j}$ to $w_{i,k}$ by reading some string $u_{k}$ if $v_jv_k$ is an edge of $G$. A simple choice for $u_1,\dots,u_n$ is the first $n$ strings when we lexicographically order the strings of length $\lceil \log_2 n\rceil$ over alphabet $\{a,b\}$. All binary paths with label $u_\ell$ of length $\lceil \log_2 n\rceil$, such that $v_jv_\ell\notin E$, from $v_{i,j}$ go to the sink state. For each $w_{i,k}$ we have a single transition, which is labeled $v_k$ and goes to state $v_{i+1,k}$, i.e., to the next gadget.
    For gadget $A_n$ only $w_{n,1}$ has a transition going out of the gadget, to state $v_{n+1,1}$.

    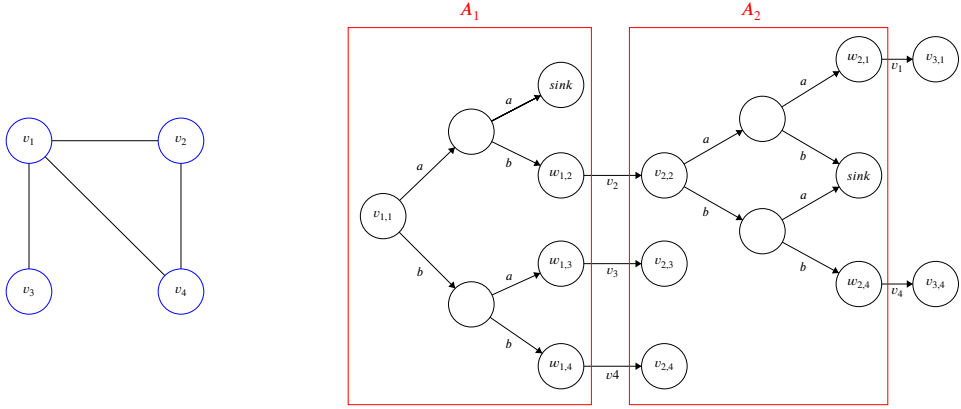
\begin{figure}
        \centering
        \scalebox{.5}{\begin{tikzpicture}[scale=0.2]
\tikzstyle{every node}+=[inner sep=0pt]

\draw [black] (-40,-20) -- (-20,-20);
\draw [black] (-40,-20) -- (-40,-40);
\draw [black] (-40,-20) -- (-20,-40);
\draw [black] (-20,-20) -- (-20,-40);

\filldraw [fill=white, draw=blue] (-40,-20) circle (3);
\draw (-40,-20) node {$v_1$};
\filldraw [fill=white, draw=blue] (-20,-20) circle (3);
\draw (-20,-20) node {$v_2$};
\filldraw [fill=white, draw=blue] (-40,-40) circle (3);
\draw (-40,-40) node {$v_3$};
\filldraw [fill=white, draw=blue] (-20,-40) circle (3);
\draw (-20,-40) node {$v_4$};

\draw [black] (20.86,-17.4) -- (27.34,-14);
\draw [black] (20.86,-17.4) -- (27.34,-14);
\draw [black] (20.86,-17.4) -- (27.34,-14);
\draw [black] (20.86,-17.4) -- (27.34,-14);

\draw [black] (30,-12.6) circle (3);
\draw (30,-12.6) node {$sink$};
\draw [black] (30,-24.6) circle (3);
\draw (30,-24.6) node {$w_{1,2}$};
\draw [black] (18.2,-18.8) circle (3);
\draw [black] (30,-49.9) circle (3);
\draw (30,-49.9) node {$w_{1,4}$};
\draw [black] (6.6,-30) circle (3);
\draw (6.6,-30) node {$v_{1,1}$};
\draw [black] (30,-36.3) circle (3);
\draw (30,-36.3) node {$w_{1,3}$};
\draw [black] (18.2,-41.7) circle (3);
\draw [black] (43.6,-24.6) circle (3);
\draw (43.6,-24.6) node {$v_{2,2}$};
\draw [black] (43.6,-36.3) circle (3);
\draw (43.6,-36.3) node {$v_{2,3}$};
\draw [black] (43.6,-49.9) circle (3);
\draw (43.6,-49.9) node {$v_{2,4}$};
\draw [black] (69.2,-9.2) circle (3);
\draw (69.2,-9.2) node {$w_{2,1}$};
\draw [black] (69.2,-39) circle (3);
\draw (69.2,-39) node {$w_{2,4}$};
\draw [black] (56.5,-17.1) circle (3);
\draw [black] (56.5,-32) circle (3);
\draw [black] (69.2,-24.6) circle (3);
\draw (69.2,-24.6) node {$sink$};
\draw [black] (79.3,-9.2) circle (3);
\draw (79.3,-9.2) node {$v_{3,1}$};
\draw [black] (79.3,-39) circle (3);
\draw (79.3,-39) node {$v_{3,4}$};
\draw [black] (20.86,-17.4) -- (27.34,-14);
\fill [black] (27.34,-14) -- (26.4,-13.92) -- (26.87,-14.81);
\draw (23.16,-15.2) node [above] {$a$};
\draw [black] (8.76,-27.92) -- (16.04,-20.88);
\fill [black] (16.04,-20.88) -- (15.12,-21.08) -- (15.81,-21.8);
\draw (11.44,-23.92) node [above] {$a$};
\draw [black] (20.89,-20.12) -- (27.31,-23.28);
\fill [black] (27.31,-23.28) -- (26.81,-22.48) -- (26.37,-23.37);
\draw (23.11,-22.21) node [below] {$b$};
\draw [black] (8.71,-32.13) -- (16.09,-39.57);
\fill [black] (16.09,-39.57) -- (15.88,-38.65) -- (15.17,-39.35);
\draw (11.88,-37.33) node [left] {$b$};
\draw [black] (20.93,-40.45) -- (27.27,-37.55);
\fill [black] (27.27,-37.55) -- (26.34,-37.43) -- (26.75,-38.34);
\draw (23.17,-38.49) node [above] {$a$};
\draw [black] (20.66,-43.41) -- (27.54,-48.19);
\fill [black] (27.54,-48.19) -- (27.16,-47.32) -- (26.59,-48.14);
\draw (23.1,-46.3) node [below] {$b$};
\draw [black] (33,-24.6) -- (40.6,-24.6);
\fill [black] (40.6,-24.6) -- (39.8,-24.1) -- (39.8,-25.1);
\draw (36.8,-25.1) node [below] {$v_2$};
\draw [black] (33,-36.3) -- (40.6,-36.3);
\fill [black] (40.6,-36.3) -- (39.8,-35.8) -- (39.8,-36.8);
\draw (36.8,-36.8) node [below] {$v_3$};
\draw [black] (33,-49.9) -- (40.6,-49.9);
\fill [black] (40.6,-49.9) -- (39.8,-49.4) -- (39.8,-50.4);
\draw (36.8,-50.4) node [below] {$v4$};
\draw [black] (46.19,-23.09) -- (53.91,-18.61);
\fill [black] (53.91,-18.61) -- (52.96,-18.58) -- (53.47,-19.44);
\draw (49.11,-20.35) node [above] {$a$};
\draw [black] (46.2,-26.09) -- (53.9,-30.51);
\fill [black] (53.9,-30.51) -- (53.45,-29.68) -- (52.96,-30.54);
\draw (49.05,-28.8) node [below] {$b$};
\draw [black] (59.05,-15.52) -- (66.65,-10.78);
\fill [black] (66.65,-10.78) -- (65.71,-10.78) -- (66.24,-11.63);
\draw (61.91,-12.65) node [above] {$a$};
\draw [black] (59.08,-18.63) -- (66.62,-23.07);
\fill [black] (66.62,-23.07) -- (66.18,-22.24) -- (65.67,-23.1);
\draw (61.85,-21.35) node [below] {$b$};
\draw [black] (59.09,-30.49) -- (66.61,-26.11);
\fill [black] (66.61,-26.11) -- (65.66,-26.08) -- (66.17,-26.95);
\draw (61.91,-27.8) node [above] {$a$};
\draw [black] (59.13,-33.45) -- (66.57,-37.55);
\fill [black] (66.57,-37.55) -- (66.11,-36.73) -- (65.63,-37.6);
\draw (61.85,-36) node [below] {$b$};
\draw [black] (72.2,-9.2) -- (76.3,-9.2);
\fill [black] (76.3,-9.2) -- (75.5,-8.7) -- (75.5,-9.7);
\draw (74.25,-9.7) node [below] {$v_1$};
\draw [black] (72.2,-39) -- (76.3,-39);
\fill [black] (76.3,-39) -- (75.5,-38.5) -- (75.5,-39.5);
\draw (74.25,-39.5) node [below] {$v_4$};

\draw (18,-4) node [above] {\color{red}\Large $A_1$};
\draw (55,-4) node [above] {\color{red}\Large $A_2$};
\filldraw [fill=none, draw=red] (2,-55) rectangle (34,-5);

\filldraw [fill=none, draw=red] (39,-55) rectangle (73,-5);
\end{tikzpicture}}
        \caption{A simple graph (without a Hamiltonian cycle) and part of the corresponding \dfawtl\ as constructed in the proof of Theorem~\ref{Thm:Prob1-dfawtl-NPhardness}. Gadgets $A_1$ and $A_2$ are marked with rectangles.}
        \label{fig:gadget}
    \end{figure}
    
    The initial state is $v_{1,1}$. There is a single final state, i.e., $v_{n+1,1}$. Fig.~\ref{fig:gadget} illustrates the construction for a simple input graph.
    The automaton accepts exactly the set of strings \[v_{i_1}\shuffle \cdots \shuffle v_{i_n}\shuffle u_{i_1}\cdots u_{i_n},\]
    where $v_{i_1}\cdots v_{i_n}$ is a walk on $G$. Therefore, $G$ is Hamiltonian if and only if $v_1\cdots v_n$ is a subsequence of some string in the language.

    The number of states and transitions in each gadget $A_i$ is $\bigo(n^2)$, since we have $n$ states $v_{i,j}$, further $n$ states $w_{i,j}$ and, for each $v_{i,j}$, an additional $\bigo(n)$ intermediate states forming the complete binary tree that starts from $v_{i,j}$ and leads to each $w_{i,k}$, in turn. As the number of gadgets is $n$, we obtain a \dfawtl\ with $\bigo(n^3)$ states, and since each state has outgoing transitions on at most $2$ letters, the total number of transitions is also $\bigo(n^3)$. The complete binary tree of depth $\lceil \log n \rceil$, having rooted paths labeled by $u_1,\dots,u_n$ can be constructed in $\bigo(n)$ time and, for each binary tree rooted in some state $v_{i,j}$, deciding which of the leaves should be $w_{i,k}$ and which should be sink states requires checking $v_jv_k\in E$. The construction is straightforward, given $G$, and it can be done in time $\poly(n,|E|)$.

Next we argue that Problem~\ref{prob:wtlsubseq} is in NP. Let us assume that the instance is a YES instance, that is, there exists some $u\in L(M)$ with $w\leq u$ and that $u$ is the shortest supersequence of $w$ accepted by $M$. Let $m'$ be the length of $w$ and $n'$ be the number of states of $M$. In what follows, when we talk about `position' we refer to a position w.r.t.\ the original input $u$, not the remaining one in any given computation step. We can describe the computation of $M$ on $u$ as the sequence of configurations 
$$(p_0, u_0) \circlearrowright^{*} (p_{i_1}, u_{i_1})\circlearrowright^{*} \cdots \circlearrowright^{*} (p_{i_{m'}}, u_{i_{m'}}) \circlearrowright^{*} (p_{i_{m'+1}},\varepsilon)$$ 
where $p_0=q_0$, $u_0=u$, $p_{i_{m'+1}}\in F$ and the position of $u$ corresponding to the $j^\text{th}$ letter of $w$ is read in step $i_j$, i.e., the $j^\text{th}$ letter of $w$ is matched immediately from configuration $(p_{i_j},u_{i_j})$. We argue that $i_{\ell+1}-i_{\ell} \leq n$, for all $\ell\in [m']$, to establish an upper bound on the length of $u$. It is straightforward that, for any $j$, the position read in the $j^\text{th}$ transition only affects the current state in transition $j$ and is not visible to the rest of the computation: since it has not been read in any step $i<j$, it was not the leftmost unread position holding a letter defined for $p_i$, and after step $j$ it does not hold any letter anymore. This means that from the input we can remove any sequence of letters read in consecutive transitions between two configurations with the same state, without affecting the acceptance of the input. Now suppose that there is some $\ell\in [m]$ such that $i_{\ell+1}-i_{\ell}>n'$. By the pigeonhole principle, the computation between those two steps goes through some state at least twice, that is, there are $1\leq j<k\leq n$ such that $p_{i_{\ell}+j}=p_{i_{\ell}+k}$. Therefore, by removing from $u$ the letters from the positions read in steps between $i_{\ell}+j$ and $i_{\ell}+k$ still results in a string $u'\in L(M)$, but $|u'|<|u|$, and since the transitions reading from positions corresponding to the letters of $w$ do not change, we still have $w\leq u'$. This, however, contradicts the assumption that $u$ is the shortest supersequence of $w$ accepted by $M$.

From here we get that $i_{m'+1}\leq (m'+1)n'$, that is, the shortest supersequence of $w$ accepted by $M$ cannot be longer than $(m'+1)n'$. This gives us a witness for $w$ being a subsequence, with size polynomial in the size of the input. Furthermore, given the witness, we can easily check whether $u\in L(M)$ in $\bigo(m'^2n'^2)$ time by running $M$ on $u$, and we can check in $\bigo(n'm')$ time whether $w\leq v$.
\end{proof}

Since our initial results deviate from the corresponding results obtained for \CFL, without suggesting that the considered problems become undecidable, completing this investigation for all other problems seems worthwhile to us. While we have excluded the approach from the general Theorem~\ref{thm:intersection_dec}, we cannot yet say anything about the approach in Theorem \ref{thm:DC_dec}. It remains an interesting open problem (also of independent interest w.r.t.\ to our research) to obtain an algorithm for computing the downward closure of a {\dfawtl}-language, or show that such an algorithm does not exist. 

While studying the problems discussed in this paper for {\dfawtl}s seems an interesting way to understand their possible further intricacies, which cause the huge gap between their status for {\CFL} and {\CSL}, respectively, another worthwhile research direction is to consider them in the context of other well-motivated classes of languages, for which all these problems are decidable, and try to obtain optimised algorithms in those cases.

\newpage
\bibliography{bib_new}

\end{document}